\documentclass[a4paper]{amsart}\usepackage[a4paper,margin=2.5cm,includefoot]{geometry}
\usepackage[utf8]{inputenc}
\usepackage[english]{babel}
\usepackage[T1]{fontenc}

\usepackage{amsmath,amssymb,mathtools,bm}
\usepackage{csquotes}
\usepackage{amstext}
\usepackage{amsthm}
\usepackage{bbm}
\usepackage{enumerate}
\usepackage[svgnames]{xcolor}
\usepackage[colorlinks,allcolors=MidnightBlue,citecolor=DarkRed]{hyperref}
\usepackage[nameinlink,capitalise]{cleveref}
\usepackage{braket}
\usepackage{dsfont}
\usepackage{color}
\usepackage{tikz}
\usepackage{pgfplots}

\makeatletter
\def\@seccntformat#1{%
	\protect\textup{\protect\@secnumfont
		\ifnum\pdfstrcmp{subsection}{#1}=0 \bfseries\fi
		\ifnum\pdfstrcmp{subsubsection}{#1}=0 \itshape\fi
		\csname the#1\endcsname
		\protect\@secnumpunct
	}%
}
\renewcommand{\@upn}{}
\DeclareRobustCommand{\crefnosort}[1]{%
	\begingroup\@cref@sortfalse\cref{#1}\endgroup
}
\makeatother

\makeatletter
\AddToHook{cmd/appendix/before}{\def\cref@section@alias{appendix}}
\makeatother

\numberwithin{equation}{section}
\newtheorem{thm}{Theorem}[section]
\newtheorem{lem}[thm]{Lemma}
\AddToHook{env/lem/begin}{\crefalias{thm}{lem}}
\newtheorem{prop}[thm]{Proposition}
\AddToHook{env/prop/begin}{\crefalias{thm}{prop}}
\newtheorem{cor}[thm]{Corollary}
\AddToHook{env/cor/begin}{\crefalias{thm}{cor}}

\theoremstyle{definition}

\AddToHook{env/defn/begin}{\crefalias{thm}{defn}}

\AddToHook{env/hyp/begin}{\crefalias{thm}{hyp}}
\renewcommand*{\thehyp}{\Alph{hyp}}

\theoremstyle{remark}
\newtheorem{rem}[thm]{Remark}
\AddToHook{env/rem/begin}{\crefalias{thm}{rem}}

\AddToHook{env/ex/begin}{\crefalias{thm}{ex}}

\crefname{hyp}{Hypothesis}{Hypotheses}\Crefname{hyp}{Hypothesis}{Hypotheses}
\crefname{lem}{Lemma}{Lemmas}\Crefname{lem}{Lemma}{Lemmas}
\crefname{thm}{Theorem}{Theorems}\Crefname{thm}{Theorem}{Theorems}
\crefname{prop}{Proposition}{Propositions}\Crefname{prop}{Proposition}{Propositions}
\crefname{enumi}{}{}\Crefname{enumi}{}{}
\creflabelformat{enumi}{#2(#1)#3}
\crefname{equation}{}{}\Crefname{equation}{}{}
\crefname{rem}{Remark}{Remarks}\Crefname{rem}{Remark}{Remarks}
\crefname{ex}{Example}{Examples}\Crefname{ex}{Example}{Examples}

\makeatletter
\renewcommand{\@upn}{} 
\makeatother

\usepackage[inline]{enumitem}
\newlist{enumthm}{enumerate}{1} 
\setlist[enumthm]{label=\upshape(\roman*),ref=\thethm\,(\roman*)}  
\crefalias{enumthmi}{thm} 
\newlist{enumcor}{enumerate}{1}
\setlist[enumcor]{label=\upshape(\roman*),ref=\thecor\,(\roman*)}
\crefalias{enumcori}{cor}
\newlist{enumlem}{enumerate}{1}
\setlist[enumlem]{label=\upshape(\roman*),ref=\thelem\,(\roman*)}
\crefalias{enumlemi}{lem}
\newlist{enumprop}{enumerate}{1}
\setlist[enumprop]{label=\upshape(\roman*),ref=\theprop\,(\roman*)}
\crefalias{enumpropi}{prop}
\newlist{enumhyp}{enumerate}{1}
\setlist[enumhyp]{label=\upshape(\roman*),ref=\thehyp\,(\roman*)}
\crefalias{enumhypi}{hyp}
\newlist{enumproof}{enumerate*}{1}
\setlist[enumproof]{label=\upshape(\roman*)}
\newlist{enumdef}{enumerate}{1}
\setlist[enumdef]{label=\upshape(\roman*),ref=\thedefn\,(\roman*)}
\crefalias{enumdefi}{defn}

\makeatletter
\newcounter{subcreftmpcnt} %
\newcommand\romansubformat[1]{(\roman{#1})} 
\def\subcref{\@ifstar\@@subcref\@subcref}
\newcommand\@subcref[2][\romansubformat]{%
	\ifcsname r@#2@cref\endcsname
	\cref@getcounter {#2}{\mylabel}%
	\setcounter{subcreftmpcnt}{\mylabel}%
	\hyperref[#2]{\romansubformat{subcreftmpcnt}}%
	\else ?? \fi}   
\newcommand\@@subcref[2][\romansubformat]{%
	\ifcsname r@#2@cref\endcsname
	\cref@getcounter {#2}{\mylabel}%
	\setcounter{subcreftmpcnt}{\mylabel}%
	\romansubformat{subcreftmpcnt}%
	\else ?? \fi}   
\makeatother

\makeatletter
\DeclareRobustCommand{\crefnosort}[1]{%
	\begingroup\@cref@sortfalse\cref{#1}\endgroup
}
\makeatother

\def\endstepsymbol{$\lozenge$}
\def\endclaimsymbol{$\lozenge$}
\newcounter{proofstep}
\AtBeginEnvironment{proof}{\setcounter{proofstep}{0}}

\crefname{proofstep}{Step}{Steps}
\Crefname{proofstep}{Step}{Steps}
\newcounter{proofclaim}
\AtBeginEnvironment{proof}{\setcounter{proofclaim}{0}}

\crefname{proofclaim}{Claim}{Claims}
\Crefname{proofclaim}{Claim}{Claims}

\newcommand{\cB}{{\mathcal B}}
\newcommand{\cF}{{\mathcal F}}
\newcommand{\cH}{{\mathcal H}}\newcommand{\cI}{{\mathcal I}}
\newcommand{\cK}{{\mathcal K}}

\newcommand{\cT}{{\mathcal T}}\newcommand{\cU}{{\mathcal U}}

\newcommand{\BC}{{\mathbb C}}

\newcommand{\BN}{{\mathbb N}}
\newcommand{\BR}{{\mathbb R}}
\newcommand{\BT}{{\mathbb T}}

\newcommand{\BZ}{{\mathbb Z}}

\usepackage{dsfont} 

\newcommand{\DSP}{{\mathds P}}

\usepackage{mathrsfs}

\newcommand{\sfd}{{\mathsf d}}

\newcommand{\sfq}{{\mathsf q}}\newcommand{\sfr}{{\mathsf r}}

\newcommand{\IN}{\BN}\newcommand{\IZ}{\BZ}\newcommand{\IR}{\BR}\newcommand{\IC}{\BC}
\newcommand{\N}{\BN}\newcommand{\R}{\BR}\newcommand{\C}{\BC}

\newcommand{\PP}{\DSP}

\newcommand{\HS}{\cH}

\newcommand{\eps}{\varepsilon}

 \renewcommand{\d}{\sfd}

\DeclareMathOperator*{\esssup}{ess\,sup}

\newcommand{\wh}[1]{\widehat{#1}}\renewcommand{\bar}[1]{\overline{#1}}

\DeclareFontFamily{U}{mathx}{\hyphenchar\font45}
\DeclareFontShape{U}{mathx}{m}{n}{
	<5> <6> <7> <8> <9> <10>
	<10.95> <12> <14.4> <17.28> <20.74> <24.88>
	mathx10
}{}
\DeclareSymbolFont{mathx}{U}{mathx}{m}{n}
\DeclareFontSubstitution{U}{mathx}{m}{n}
\DeclareMathAccent{\widecheck}{0}{mathx}{"71}
\DeclareMathAccent{\wideparen}{0}{mathx}{"75}

\newcommand{\chr}{\mathbf 1}

\title[Optimal Convergence Rate of 1D Quantum Walks]{On the Optimal Rate of Convergence\\ for Translation-Invariant 1D Quantum Walks}
\author{Benjamin Hinrichs}
\author{Pascal Mittenb\"uhler}
\address{Universit\"at Paderborn, Institut f\"ur Mathematik, Institut f\"ur Photonische Quantensysteme, Warburger Str. 100, 33098 Paderborn, Germany}
\email{benjamin.hinrichs@math.upb.de, mittenbu@math.upb.de}

\usepackage{cancel}

\newcommand{\tr}{\operatorname{tr}}

\begin{document}

\begin{abstract} 
	\noindent
	We study the convergence rate of translation-invariant discrete-time quantum dynamics on a one-dimensional lattice. We prove that the cumulative distributions function of the ballistically scaled position $X({n})/{n}$ after $n$ steps converges at a rate of ${n}^{-1/3}$ in the L\'evy metric as ${n}\to\infty$. In the special case of shift-coin quantum walks with two-dimensional coin space, we recover the same convergence rate for the supremum distance and prove optimality.
\end{abstract}

\maketitle

\section{Introduction}

Quantum walks have since their introduction \cite{Ambainisetal.2001} attracted a lot of attention as an experimentally realizable platform \cite{Peretsetal.2008,Karskietal.2009,Schreiberetal.2010,Genskeetal.2013} with many applications, e.g., for search algorithms in quantum computing \cite{Ambainis.2007,Lovettetal.2010,Portugal.2013,Montanaro.2016} or quantum simulation \cite{AspuruGuzikWalther.2012,Childs.2009}.
Particular emphasis has been put on the analogy of quantum walks to the classical random walk, where the main computational advantage of the quantum walk is the ballistic propagation of information.
Especially, analogous results to the classical central limit theorem have by now been established in various settings with versatile mathematical approaches, see for example \cite{Konno.2002,GrimmettJansonScudo.2004,Konno.2005b,AhlbrechtVogtsWernerWerner.2011,SunadaTate.2012,Suzuki.2016,Wada.2020} and references therein.

Surprisingly, apart from the recent work \cite{CedzichJoyeWerner2.2025}, global error bounds for the approach to the asymptotic position distribution in the setting of quantum walks have not been studied, to our knowledge. In that article, the authors prove exponential decay of quantum information outside of the convex hull of the propagation region.
A local error bound, {except for a thin layer around the boundary of the propagation region}, was further derived in \cite{SunadaTate.2012}.
The absence of global error estimates stands in strong contrast to the celebrated Berry--Esseen theorem for the classical random walk \cite{Berry.1941,Esseen.1945}. In this article, we fill this gap and prove a Berry--Esseen type theorem for one-dimensional translation-invariant quantum walks.

In comparison to the classical random walk, which has a convergence order of $n^{-1/2}$ as the number of steps $n$ goes to $\infty$, we in fact discover that quantum walks converge slower at the order of $n^{-1/3}$. This is due to the fact that most information is located directly in the ballistically propagating wavefront region, where the convergence is thus slowed down. In view of the many applications of quantum walks in quantum computing and their experimental realization, our results will provide an important ingredient in estimating computational errors in the future. Natural extensions of our result would also be to study non translation-invariant walks or higher-dimensional systems.
As far as one-dimensional systems go, our upper bound on the convergence rate in fact holds for a large class of unitary lattice dynamics, especially including the so-called split-step quantum walks, which can be described in the framework of CMV matrices \cite{BourgetHowlandJoye.2003,CanteroMoralVelazquez.2003,CanteroGrunbaumMoralVelazquez.2012,CedzichFillman.2024}.

\smallskip
For the purpose of this introduction, let $X_n$ and $X_\infty$ denote the position after $n$ time steps and asymptotic position of a walker on the one-dimensional lattice $\IZ$ and let the corresponding probability distributions be $\PP\!_\sfr$ and $\PP\!_\sfq$, respectively. Then the Berry--Esseen theorem states that 
\begin{align*}
	\sup_{x\in \IR} \big|\PP\!_\sfr[X_n/\sqrt n\le x] - \PP\!_\sfr[X_\infty\le x]\big| \asymp n^{-\frac 1 2 }, \qquad n\to\infty.
\end{align*}
The main result of this article is the analogous statement
\begin{align*}
	\sup_{x\in \IR} \big|\PP\!_\sfq[X_n/n\le x] - \PP\!_\sfq[X_\infty\le x]\big| \asymp n^{-\frac 1 3 }, \qquad n\to\infty.
\end{align*}
We further prove a corresponding upper bound for general discrete-time quantum dynamics on the one-dimensional lattice in the L\'evy metric, which is in fact an important intermediate step in the proof of the sup-norm bound as well.

Our proof essentially consists of two parts: {First}, we prove an adapted Esseen inequality, estimating differences of cumulative distribution functions from the corresponding characteristic functions. Applying it to the setting of 1D quantum walks yields the  convergence rate $n^{-1/3}$ for the L\'evy distance. We then carefully analyze the concrete position distribution for the shift-coin quantum walk, by building on results from \cite{SunadaTate.2012}. In this second step, we especially need to focus on a treatment of the wavefront region mentioned above, which eventually yields the optimal convergence rate in the supremum distance.

\subsubsection*{Structure of the Article} 
In \cref{sec:int}, we present the exact setting of this article and state our main results \cref{thm:main,thm:quantum_rate_supnorm}. In \cref{sec:esseen}, we then prove a generalized Esseen--Zolotarev type inequality, which applied to our setting yields the proof of \cref{thm:main}. In the final \cref{sec:coinstep}, we then sharpen the estimates for shift-coin quantum walks to prove \cref{thm:quantum_rate_supnorm}.

\subsubsection*{Acknowledgments}
PM thanks Dr. Maik Reddiger for reminding him of the importance of the article \cite{SunadaTate.2012}, Prof. Tadahiro Miyao for guidance and Prof. Itaru Sasaki for guidance and valuable discussions.
PM was supported by a fellowship of the German Academic Exchange
Service (DAAD) and thanks Shinshu University and in particular Prof. Itaru Sasaki for their hospitality during his stay.
Both authors acknowledge funding by the Ministry of Culture and Science of the
State of North Rhine-Westphalia within the project `PhoQC' (Grant
Nr. PROFILNRW-2020-067).

\section{Model and Results}

Let us now introduce the models under consideration and state the precise results.
We will first introduce the abstract concept of discrete time one-dimensional lattice dynamics in \cref{sec:int} and recall the context of
the well-known central limit theorems in the translation-invariant case.
We will then state our general result on this type of dynamics, which is a Berry--Esseen type theorem
for the cumulative distribution functions converging in L\'evy metric in the order $n^{-1/3}$ as $n\to\infty$, with $n$ being the number of time steps.
Then, on the concrete example of two-dimensional local Hilbert spaces,
we prove that this rate of convergence is optimal even in supremum norm, also referred to as Kolmogorov metric in this context.

\subsection{Lattice Dynamics}\label{sec:int}
Let us start by introducing the general setting of this article.
\subsubsection{General Setup}
Let $\cK$ be a for now arbitrary (local) Hilbert space and
as global Hilbert space consider $\cH = \ell^2(\IZ;\cK)$, i.e., the space of square-summable sequences taking values in $\cK$.
As states of the so-described quantum system, we consider density matrices, i.e., trace-class positive operators $\rho\in\cT(\cH)$ satisfying $\tr(\rho)=1$.
The time-step operator $W\in\cU(\cH)$ describes the time evolution under a single (discrete) time step, so the system initialized in a state $\rho$ will have state $W^n\rho W^{-n}$ after $n$ steps.
Given a selfadjoint operator $A$ on $\cH$, the probability to measure $A$ for the system in state $\rho$ in a set $M\subset\IR$ is $\tr(\rho\chr_M(A))$. This allows us to define the corresponding {\em cumulative distribution functions} (CDF) by
\begin{align}
	\label{eq:CDF}
	F_A^\rho (x) \coloneqq \tr\big(\rho \chr_{(-\infty,x]}(A)\big).
\end{align}
We are here most interested in choosing $A=X$ the (selfadjoint) position operator given by $(X \psi)_k = k\psi_k$, $\psi\in\cH$.
Then for given initial state $\rho$, we analyze the ballistically rescaled position CDF.
After $n$ steps, it is given by
$F_{X/n}^{W^n \rho W^{-n}} = F^\rho_{X_n}$
with $X_n \coloneqq W^{-n}X W^n/n$.

\subsubsection{Translation-Invariance and Weak Limit Theorems}
We will restrict our attention to translation-invariant time step operators $W$, i.e., assume that $W$ commutes with the right shift operator $(T\psi)_k \coloneqq \psi_{k-1}$.
Denoting by $\cF:\ell^2(\IZ;\cK)\to L^2(\BT;\cK)$ with $\BT=[0,2\pi)$ the usual unitary Fourier transform defined by
\begin{align*}
	(\cF\psi)(p) = \frac{1}{\sqrt{2\pi}}\sum_{k\in\IZ}\psi_ke^{ikp}, \quad p\in\BT,
\end{align*}
translation-invariance is equivalent to the existence of a map $\wh W:\BT\to\cU(\cK)$ such that
\begin{align*}
	\big(\cF W \cF^* \psi\big)(p) = \wh W(p)\psi(p).
\end{align*}
We will here further assume that $\wh W(p)$ has purely discrete spectrum for all $p\in\BT$, i.e., there exists a (finite or countable) family of orthogonal projections $(\Pi_k(p))_{k\in \cI}$ on $\cK$ and a corresponding family of real numbers $(\omega_k(p))_{k\in\cI}\subset \IR$ such that
\begin{align}
	\label{eq:Wp}
	\wh W(p) = \sum_{k\in\cI} e^{i \omega_k(p)}\Pi_k(p).
\end{align}
If $\omega_k$ is differentiable in almost every $p\in\BT$, we can then define the so-called velocity operator as the selfadjoint operator given by
\begin{align}
	\label{eq:V}
	\cF V\cF^* (p) = \sum_{k\in\cI} \omega_k'(p)\Pi_k(p).
\end{align}
In this case the strong convergence $e^{i t X_n} \xrightarrow{n\to\infty} e^{it V}$ for all $t\in\IR$ is well-established as a type of weak central limit theorem for quantum walks \cite{Konno.2002,GrimmettJansonScudo.2004,AhlbrechtVogtsWernerWerner.2011,Suzuki.2016}.
By the Portmanteau theorem, this is equivalent to the convergence
\begin{align*}
	F^\rho_{X_n}(x) \xrightarrow{n\to\infty} F^\rho_V(x)
\end{align*}
in all points of continuity $x$ of $F^\rho_V$.
The aim of this article is to explicitly quantify this convergence in appropriately chosen metrics for the CDFs under consideration.

\subsection{L\'evy Metric and General Upper Bound}
There are various natural metrizations in terms of CDFs for weak convergence of probability measures existing in the literature, see for example the review \cite{Bobkov.2016}.
Especially for not everywhere differentiable CDFs, the supremum metric is not very well suited. It is thus natural to consider the  L\'evy metric given by
\begin{align}
	\label{eq:Levy}
	L\left(F,G \right)\coloneqq \sup_{x\in\IR}\inf\{\epsilon>0: F(x-\epsilon)-\epsilon\leq G(x)\leq F(x+\epsilon)+\epsilon\}.
\end{align}
We note that the L\'evy metric is upper-bounded by the supremum metric\footnote{
{\em Proof of \cref{lem:linfty}.}
		Setting $\epsilon\coloneqq\|F-G\|_\infty $ and using that CDFs are increasing, we can estimate for any $x\in\IR$
		\begin{align*}
			F(x- \epsilon)-\epsilon&\leq F(x)-\lvert F(x)-G(x)\rvert \leq G(x) \leq F(x+\epsilon)-F(x)+G(x) \leq F(x+\epsilon) + \epsilon.\hfill\qed
		\end{align*}
}
\begin{equation}
	\label{lem:linfty}
	L(F,G)\leq \sup_{x\in\IR}\lvert F(x)-G(x)\rvert \eqqcolon \|F-G\|_\infty,
\end{equation}
In fact, in regions where one of the CDFs is continuously differentiable, the L\'evy metric vice versa can also yield an upper bound on the supremum metric,
as we will make more precise in \cref{lem:smooth} below.

Under very mild additional assumptions on $W$, we can identify the convergence rate $n^{-1/3}$ in terms of the L\'evy metric.
\begin{thm}\label{thm:main}
    In \cref{eq:Wp} additionally assume that $\omega_k\in C^2(\BT;\IR)$ and $\Pi_k\in C^1(\BT;\cB(\HS))$ such that 
    \[\sum_{k\in\cI}\sup_{p\in\BT} \|\Pi_k'(p)\|_{\cB(\cK)} <\infty \quad \mbox{and}\quad\sup_{k\in\cI,\,p\in\BT}|\omega_k''(p)|<\infty. \]
     Then, for any density matrix $\rho$ satisfying 
    $
        \tr(\lvert X\rvert \rho)<\infty
    $, 
    there exists a constant $C>0$ such that
	\[L(F_{X_n}^{\rho},F^{\rho}_V)\le C  n^{-\frac 1 3 }.\]
\end{thm}
\begin{proof}
	This follows by combining \cref{lem:ZolotarevSim,Cor:triangle}.
	Details can be found in the end of \cref{sec:esseen}.
\end{proof}

\subsection{Kolmogorov Metric and Optimality for Step-Coin Walks}
\label{subsec:Kolmo}
Let us now introduce the concept of the usual shift-coin quantum walk with the two-dimensional coin space $\cK=\IC^2$. In this setting, recalling the right shift operator $T$ from above and using the natural identification $\HS =\ell^2(\IZ;\IC^2)\cong \ell^2(\IZ)\oplus\ell^2(\IZ)$,
the {\em step operator} $S$ is defined as
\begin{align*}
	S(\psi_1\oplus\psi_2) \cong T\psi_1\oplus T^{-1}\psi_2.
\end{align*}
Note that $\cF S\cF^*(p) = e^{ip}\oplus e^{-ip}$, i.e., $\cF S\cF^*$ is the direct sum of multiplication operators.
We will then call $W$ the time-step operator of a shift-coin walk on $\cH$ if $W=SC$ for some $C\in\cU(\IC^2)$ acting as $(C\psi)_k=C\psi_k$, $k\in\IZ$ called {\em coin operator}.

Our next result proves that in this setting the rate of convergence in \cref{thm:main} carries over to the supremum metric, also known as Kolmogorov metric in this context, and is optimal.
\begin{thm}
    \label{thm:quantum_rate_supnorm}
    Assume $C \in \mathcal{U}(\mathbb{C}^2)$ has only non-zero entries and assume that the density matrix $\rho$ is a finite sum of orthogonal projections of the form $\ket{\delta_n\phi}\bra{\delta_n\phi}$, where $\phi\in\IC^2$ and $\delta_n\in\ell^2(\IZ)$, is as usually given by $(\delta_n)_k=1$ if $n=k$ and $(\delta_n)_k=0$ else.
    Then there exist $C_1,C_2>0$ such that
    \[
        C_1n^{-\frac 1 3 }\le \bigl\| F_{X_n}^{\rho} - F_V^\rho \bigr\|_\infty \le C_2n^{-\frac 1 3 }.
    \]
    \end{thm}
    \begin{proof}
    	The proof is given in \cref{sec:coinstep}.
    \end{proof}
    \begin{rem}
        In contrast to the classical Berry--Esseen rate, which is $n^{-1/2}$, the one-dimensional shift-coin quantum walk converges strictly
        slower.  The reduced rate arises from the formation of ballistic wavefronts:
        near the edges of the support of its density, $F_V^\rho$ develops a nondifferentiable cusp
        and the approximation error is dominated by this boundary behaviour.
        Analyzing this boundary region is the main ingredient to the proof of the above theorem.
    \end{rem}
    \begin{rem}
    	\label{rem:lowerboundLevy}
    	It would be desirable to also obtain a corresponding lower bound for the L\'evy metric in the shift-coin quantum walk case.
    	Our method in this case yields that for all $\eps>0$
    	\[ n^{-\frac23-\eps}\lesssim L(F_{X_n^\rho},F_V^\rho)\lesssim n^{-\frac13 }, \]
    	but matching upper and lower bounds would require a more subtle analysis of the wavefront sector, see \cref{rem:lowerboundLevyproof} for the simple proof.
    \end{rem}
\section{Esseen--Zolotarev Type Estimates}
\label{sec:esseen}
In this \lcnamecref{sec:esseen}, we prove \cref{thm:main},
by appropriately generalizing the methods used for Berry--Esseen type theorems and subsequently applying it to the lattice dynamic setting.

Throughout this \lcnamecref{sec:esseen}, we will work with {\em cumulative distribution functions} (CDF) $F:\IR\to[0,1]$ of a Borel probability measure $\mu_{{F}}$ on $\IR$,
i.e., $F(x)=\mu_{{F}}((-\infty,x])$. That is, $F$ is increasing, right-continuous and satisfies $\lim_{x\to-\infty}F(x)=0$, $\lim_{x\to+\infty}F(x)=1$.
We will denote the corresponding characteristic functions by
\[
	\wh F(\lambda) \coloneqq \int_\IR e^{i \lambda x}\d\mu_{{F}}(x). 
\]

\subsection{A Generalized Zolotarev Inequality}
\label{subsec:Zolo}
In the original work of Berry and Esseen \cite{Berry.1941,Esseen.1945},
the supremum distance of two CDFs was estimated in terms of the characteristic function of the probability measure.
Later generalization and improvements were also given, e.g., in \cite{Fainleib.1968,BentkusGotze.1996}, also see \cite[\S\,2]{Bobkov.2016} for an extensive discussion.
Nevertheless, none of these strategies will provide the optimal rate of convergence in our setting, due to the irregularity of the CDFs or more precisely their densities.
Thus, instead we consider here the L\'evy metric introduced in \cref{eq:Levy}, which provides a metrization of the weak convergence of probability measures.
There also exist many estimates on the L\'evy metric in terms of the characteristic functions of the corresponding measures, e.g, \cite{Bohman.1961,Zolotarev.1971},
again see \cite[\S\,3]{Bobkov.2016} for an overview,
but none in the form presented in those works provides exactly the optimal convergence rate derived here.
We therefore build on the work by Zolotarev \cite{Zolotarev.1971} and prove the following estimate.
\begin{thm}\label{lem:ZolotarevSim}
    There exists a global constant $C>0$ such that for CDFs $F$, $G$ and any $\epsilon\in(0,1]$ 
    \begin{align*}
        L(F,G)\leq \epsilon + \epsilon^{-2}C\max\left(\sup_{\lambda\in(0,1]}\left\lvert\frac{\wh F(\lambda)-\wh G(\lambda)}{\lambda}\right\rvert,\sup_{\lambda\in(1,\infty)}\left\lvert\frac{\wh F(\lambda)-\wh G(\lambda)}{\lambda^2}\right\rvert\right).
    \end{align*}
\end{thm}
Our proof closely follows the lines of Zolotarev's original proof in \cite{Zolotarev.1971}.
The starting point is the following smoothing inequality, which we prove here for completeness.
In the statement we use the usual convolution of CDFs $F$ and $G$ with corresponding measures $\mu_F$ and $\mu_G$ on $\IR$ defined by
\begin{align}
	\label{def:conv}
	F\ast G (x) \coloneqq \int_\IR F(x-y)\d\mu_G(y) = \int_\IR G(y)\d\mu_F(x-y).
\end{align}
\begin{lem}\label{lem:smoothing}
    For CDFs $F,G,H$ on $\IR$ and any $\epsilon>0$, it holds that
    \begin{equation*}
        0\leq L(F,G)- L(F\ast H,G\ast H)\leq \max\{\epsilon, 1-H(\epsilon/2)+H(-\epsilon/2)\}.
    \end{equation*}
\end{lem}
\begin{proof}
    From the definition of the L\'evy metric, for any $R > L(F,G)$ and any $x\in\IR$, we have
    \begin{equation*}
        F(x- R)-R\leq G(x)\leq F(x+ R)+R.
    \end{equation*}
    For fixed $y\in\IR$, integrating against the measure corresponding to $H$ as in \cref{def:conv} yields
    \begin{equation*}
        F\ast H(y- R)-R\leq G\ast H(y)\leq F\ast H(y+ R)+R.
    \end{equation*}
    This implies
    \begin{equation*}
        R\geq L(F*H,G*H)\eqqcolon S.
    \end{equation*}
    Taking the infimum over all $R> L(F,G)$ proves the lower bound in the statement.
    
    For the upper bound, we will first show the following inequalities for any $\epsilon>0$ and $x\in\R$:
    \begin{equation}\label{eq:conv}
        F(x- \epsilon)-1+H( \epsilon)\leq F*H(x)\leq F(x+ \epsilon)+H(- \epsilon).
    \end{equation}
    To see this we introduce two independent random variables $X$ and $Y$ on the probability space $(\Omega,\mathscr{F},\PP)$ distributed according to $F$ and $H$, i.e., $\PP[X\le x] = F(x)$, $\PP[Y\le y]=H(y)$.
    Then, the convolution can also be written and estimated as
    \begin{align*}
    	F\ast H(x)&=   \PP [X+Y\leq x]\\
    	&\geq    \PP [X\leq x- \epsilon,\, Y\leq  \epsilon]\\
    	&\overset{\textup{indep.}}{=}   \PP [X\leq x- \epsilon]   \PP [Y\leq \epsilon]\\
    	&= F(x- \epsilon)H( \epsilon)\\
    	&= F(x- \epsilon)(1-1+H( \epsilon))\\
    	&\overset{\textup{monot.}}{\geq}F(x- \epsilon)-1+H( \epsilon),
    \end{align*}
    proving the lower bound in \cref{eq:conv}.
    For the upper bound, we use
    \begin{align*}
    	\PP [X+Y\leq x]& \le   \PP [X\leq x +\epsilon]+   \PP [X> x+\epsilon,\, Y\leq- \epsilon]\\
    	&\leq       \PP [X- \epsilon\leq x] + \PP [Y\leq - \epsilon]\\
    	&= F(x+ \epsilon) + H(- \epsilon),
    \end{align*}
    thus finishing the proof of \cref{eq:conv}.
    
    The upper bound in the statement thus follows, by estimating
    \begin{align*}
    	G(x)&\overset{\eqref{eq:conv}}{\leq} G\ast H(x+ \epsilon)+1-H(  \epsilon)\\
    	&\leq F\ast H(x+ (\epsilon+S))+S+1-H( \epsilon)\\
    	&\overset{\eqref{eq:conv}}{\leq} F(x+ (2\epsilon+S))+\left(1-H( \epsilon)+H(- \epsilon)+S\right)
    \end{align*}
    as well as
    \begin{align*}
    	G(x)&\overset{\eqref{eq:conv}}{\geq}G*H(x- \epsilon)+H( \epsilon)\\
    	&\geq F(x- (\epsilon +S))+H( \epsilon)-S\\
    	&\overset{\eqref{eq:conv}}{\geq} F(x- (2\epsilon +S))-1+H( \epsilon)-H(- \epsilon)-S\\
    	&=F(x- (2\epsilon +S))-\left(1-H( \epsilon)+H(- \epsilon)+S\right).
    \end{align*}
    Let $M\coloneqq S+ \max\left(2\epsilon,1-H(\epsilon)-H(-\epsilon)\right)$ then monotonicity of $F$ implies for both inequalities
    \begin{align*}
    	F(x-M)-M\leq G(x)&\leq F(x+M)+M
    \end{align*}
    and therefore $L(F,G)\leq M$. This finishes the proof, since $\epsilon>0$ was arbitrary.
\end{proof}
We can apply this to give the
\begin{proof}[Proof of \cref{lem:ZolotarevSim}]
    In the following, we will convolve $F$ and $G$ with a family of smoothing CDFs,
    which we choose as sums of independent, identically distributed random variables with a triangular density.
    More precisely,
    for $\epsilon>0$ and $n\in\IN$, we define the densities
	\begin{align*}
	    \theta_\epsilon^n(x)\coloneqq \begin{cases}
	        0,\quad &\textup{for } \lvert x\rvert \geq\frac{\epsilon}{2n}\\
	        \left(\frac{2n}{\epsilon}\right)^2\left(\frac{\epsilon}{2n}-\lvert x\rvert\right),\quad &\textup{else }. 
	    \end{cases}
	\end{align*}
This density is supported on the interval $[-\frac{\epsilon}{2n},\frac{\epsilon}{2n}]$ and thus the density of the sum of $n$ independent and identically distributed random variables is supported in the interval $[-\frac{\epsilon}{2},\frac{\epsilon}{2}]$. The cumulative distribution function of the sum is now given by 
\begin{align*}
    \Theta_\epsilon^n(x)=\int_{-\infty}^x \left(\ast^n\theta_\epsilon^n\right)(y)\d y.
\end{align*}
A direct calculation yields the characteristic function
    \[\widehat{\Theta^n_\epsilon}(\lambda)=\left(\frac{\sin\left(\frac{\epsilon\lambda}{2n}\right)}{\frac{\epsilon\lambda}{2n}}\right)^n .\]
Since the density of our smoothing random variable is supported on $[-\frac\epsilon2,\frac\epsilon 2]$, we have $\Theta_\epsilon^n(-\epsilon/2)=0$ and  $\Theta_\epsilon^n(\epsilon/2)=1$,
so \cref{lem:smoothing} yields

\begin{align}\label{eq:smoothlevy}
    L(F,G)\leq \epsilon+ L(F*\Theta_\epsilon^n,G*\Theta_\epsilon^n)\leq \epsilon + \sup_{x\in\R}\left(F*\Theta_\epsilon^n(x)-G*\Theta_\epsilon^n(x)\right),
\end{align}
where we applied \cref{lem:linfty} in the second step. We from now on set $n=3$.

    For CDFs $F$ and $G$ the convolution with a distribution that has a continuous density with respect to the Lebesgue measure,
    leaves $F\ast \Theta_\epsilon^3$ and $G\ast\Theta_\epsilon^3$ continuous and thus the Fourier inversion formula yields the estimate 
    \begin{align*}
        \left\lvert F\ast\Theta_\epsilon^3(x)-G\ast\Theta_\epsilon^3(x)\right\rvert&=\frac{1}{2\pi}\left\lvert\int_{\IR}e^{-i\lambda x} \frac{\widehat{F\ast\Theta_\epsilon^3}(\lambda)-\widehat{G\ast\Theta_\epsilon^3}(\lambda)}{\lambda}\d \lambda\right\rvert\\
        &\leq \frac{1}{\pi}\int_{0}^{\infty}\lvert \widehat{F}(\lambda)-\wh{G}(\lambda) \rvert\left\lvert\frac{\wh{\Theta_\epsilon^3}(\lambda)}{\lambda}\right\rvert\d \lambda.
    \end{align*}
    We now introduce the function $\psi_\epsilon:\IR\setminus\{0\}\to\IR$ defined by
    \begin{align*}
        \psi_\epsilon(\lambda)\coloneq\begin{cases}
            \lambda^{-1}\epsilon^{-1},\quad \textup{for } \lambda\leq 1\\
            \lambda^{-2}\epsilon^{-2}, \quad \textup{for } \lambda>1.
        \end{cases}
    \end{align*}
	and further estimate
    \begin{align*}
        &\frac{1}{\pi}\int_{0}^{\infty}\lvert \wh{F}(\lambda)-\wh{G}(\lambda) \rvert\left\lvert\frac{\wh{\Theta_\epsilon^3}(\lambda)}{\lambda}\right\rvert\d \lambda=  \frac{1}{\pi}\int_{0}^{\infty}\psi_\epsilon(\lambda)\lvert \wh{F}(\lambda)-\wh{G}(\lambda) \rvert\left\lvert\frac{\wh{\Theta_\epsilon^3}(\lambda)}{\psi_\epsilon(\lambda)\lambda}\right\rvert\d \lambda\\
        &\qquad \leq \frac{C}{\pi}\max\left(\sup_{\lambda\in(0,1]}\left(\frac{\lvert \wh{F}(\lambda)-\wh{G}(\lambda)\rvert }{\epsilon\lambda}\right),\sup_{\lambda\in(1,\infty)}\left(\frac{\lvert \wh{F}(\lambda)-\wh{G}(\lambda)\rvert}{\epsilon^2\lambda^2}\right)\right),
    \end{align*}
    where
    \[C\coloneq \int_0^\infty \left\lvert\frac{\wh{\Theta_1^3}(\lambda)}{\psi_1(\lambda)\lambda}\right\rvert\d \lambda<\infty. \]
    We now use $\epsilon^{-2}\geq \epsilon^{-1}$ for $\epsilon\in(0,1]$ to see 
    \begin{align*}
        \left\lvert F\ast\Theta_\epsilon^3(x)-G\ast\Theta_\epsilon^3(x)\right\rvert&\leq \frac{C}{\epsilon^2\pi}\max\left(\sup_{\lambda\in(0,1]}\left(\frac{\lvert \wh{F}(\lambda)-\wh{G}(\lambda)\rvert }{\lambda}\right),\sup_{\lambda\in(1,\infty)}\left(\frac{\lvert \wh{F}(\lambda)-\wh{G}(\lambda)\rvert}{\lambda^2}\right)\right).
    \end{align*}
    Since the right hand side is independent of $x$, we can now take the supremum over all $x\in\IR$ and combining with \cref{eq:smoothlevy} proves the claim. 
\end{proof}

\subsection{Estimates on The Characteristic Functions}
\label{subsec:charfunc}
We now want to estimate the convergence of the CDFs for unitary 1D-lattice dynamics given by a time step operator $W$.
In view of \cref{lem:ZolotarevSim}, this can be done by estimating the difference of the corresponding characteristic functions, which in view of the spectral theorem can be evaluated as
\begin{align*}
	\wh F_{X_n}^{\rho}(\lambda) = \tr \big( \rho W^{-n}e^{i\lambda X}W^n\big), \qquad \wh F_V^\rho(\lambda) = \tr\big(\rho e^{i \lambda V}\big).
\end{align*}
The main result on the level of characteristic functions now is
\begin{thm}\label{Cor:triangle}
	Let $W\in\cU(\cH)$ satisfy the assumptions of \cref{thm:main} and let $\rho\in\cT(\cH)$ be a density matrix satisfying $\tr(\lvert X\rvert\rho )<\infty$. Then, for all $n\in\N$ and $\lambda\in\R$,
	\begin{equation*}
		\left\lvert \wh F_{X_n}^{\rho}\left(\frac{\lambda}{n}\right)-\wh F_V^\rho(\lambda)\right\rvert\leq \frac{\lvert\lambda\rvert^2}{n}\sup_{k\in\cI,p\in\BT}\lvert\omega_k''(p)\rvert+ \frac{\lvert\lambda\rvert}{n}\left(\tr( \lvert X\rvert\rho )+\left(\sum_{k\in\cI}\sup_{p\in\BT}\lVert  \Pi _k'(p)\rVert\right)\right).
	\end{equation*}
\end{thm}
\begin{proof}
	By the triangle inequality and the linearity of the trace, we have 
	\begin{align*}
		\left\lvert \wh F_{X_n}^{\rho}\left(\frac{\lambda}{n}\right)-\wh F_V^\rho(\lambda) \right\rvert\leq \left\lvert \tr\left(\rho W^{-n}e^{i\frac{\lambda}{n}{X}}W^n\left(1-e^{-i\frac{\lambda}{n}{X}}\right)\right)\right\rvert+\left\lvert \tr\left(\rho\left(e^{i\lambda V}-W^{-n}e^{i\frac{\lambda}{n}{X}}W^n e^{-i\frac{\lambda}{n}{X}}\right)\right) \right\rvert.
	\end{align*}
	These terms are estimated separately in the below \cref{lem:simpleest,lem:triangle}.
\end{proof}
The first estimate is simple in view of the spectral theorem.
\begin{lem}
	\label{lem:simpleest}
   	Given a density matrix $\rho\in\cT(\HS)$ and any unitary $W\in\cU(\HS)$, 
     for all $n\in\N$ and $\lambda\in\IR$, we have
    \begin{align*}
        \left\lvert\tr\left({W^{-n}}e^{i\frac{\lambda}{n} {X}}W^n\rho -{W^{-n}}e^{i\frac{\lambda}{n} {X}}W^ne^{-i\frac{\lambda}{n} X}\rho \right)\right\rvert\leq\frac{\lvert\lambda\rvert}{n}\tr\left(\lvert {X}\rvert \rho \right),
    \end{align*}
    where the right hand side might be infinite.
\end{lem}
\begin{proof}
    Since $W^{-n}e^{i\frac{\lambda}{n}{X}}W^n$ is unitary, we have
    \begin{align*}
        \left\lvert\tr\left({W^{-n}}e^{i\frac{\lambda}{n} {X}}W^n\rho -{W^{-n}}e^{i\frac{\lambda}{n} {X}}W^ne^{-i\frac{\lambda}{n} {X}}\rho \right)\right\rvert\leq \left\lvert\tr\left((1-e^{i\frac{\lambda}{n}{X}})\rho \right)\right\rvert.
    \end{align*}
    As $\rho $ is a density matrix, it is a positive trace-class operator and thus can be written as
    $
        \rho =\sum_i \lambda_i \ket{\psi_i}\bra{\psi_i}.
    $
    with $\sum_i \lambda_i=1$ and $\lambda_i\ge 0$ for all $i$.
    From the simple inequality
    $\lvert 1-e^{ix}\rvert\leq \lvert x\rvert$ for $x\in\IR$,
    it thus follows that
    \begin{align*}
        \left\lvert\tr\left(\big(1-e^{i\frac{\lambda}{n}{X}}\big)\rho \right)\right\rvert&\leq \sum_i\sum_{x\in\IZ} \lambda_i \left\lvert1-e^{i\frac{\lambda}{n}x}\right\rvert \|\psi_i(x)\|_\cK^2
        \leq \sum_i\sum_{x\in\IZ} \lambda_i \left\lvert\frac{\lambda}{n}x\right\rvert \|\psi_i(x)\|_\cK^2
        = \frac{\lvert \lambda\rvert}{n}\tr\left(\lvert {X}\rvert \rho \right).\qedhere
    \end{align*}
\end{proof}

\begin{lem}\label{lem:triangle}
    Let $\rho\in\cT(\cH)$ be a density matrix, and $W\in\cU(\HS)$ be translation-invariant with locally purely discrete spectrum, i.e., \cref{eq:Wp} holds, and
	assume that   $\omega_k\in C^2(\BT;\IR)$ and $\Pi_k\in C^1(\BT;\cB(\HS))$.
	Further, define the asymptotic velocity operator by \cref{eq:V} as before.
     Then, for any $n\in\IN$ and $ \lambda\in\IR$,
    \begin{align*}
        \left\lvert  \tr\left(\rho\left(e^{i\lambda V}-W^{-n}e^{i\frac{\lambda}{n}{X}}W^n e^{-i\frac{\lambda}{n}{X}}\right)\right) \right\rvert\leq  \frac{\lvert\lambda\rvert^2}{n}\sup_{p\in\BT,k\in\cI}\lvert\omega_k''(p)\rvert+ \frac{\lvert\lambda\rvert}{n}\sum_{k\in\cI}\sup_{p\in\BT}\lVert \Pi _k'(p)\rVert_{\cB(\cK)},
    \end{align*} 
	where again the right hand side might be infinite.
\end{lem}
  \begin{proof}
  	Since $\tr(\rho A)\le \|A\|_{\cB(\HS)}{\cdot}\tr(\rho)=\|A\|_{\cB(\HS)}$,
  	it suffices to estimate the rate of convergence in  
    \begin{align*}
        W^{-n}e^{i\frac{\lambda}{n}{X}}W^n e^{-i\frac{\lambda}{n}{X}}\xrightarrow{n\to\infty} e^{i\lambda V}
    \end{align*}
    with respect to the operator norm.
    We note that by assumption, both operators in the difference act locally after Fourier transforming, i.e., for $p\in\BT$, we have as operator identities on $\cK$
    \begin{align*}
    	&\cF W^{-n}e^{i\frac{\lambda}{n}{X}}W^n e^{-i\frac{\lambda}{n}{X}} \cF^* (p) = \sum_{k,\ell\in\cI}e^{i n(\omega_\ell(p+\tfrac\lambda n)-\omega_k(p))}\Pi_k(p)\Pi_\ell\left(p+\frac\lambda n\right),
    	\\
    	& \cF e^{i \lambda V}\cF^* (p) = \sum_{k\in\cI}e^{i\lambda\omega_k'(p)}\Pi_k(p).
    \end{align*}
    We now introduce the intermediate operators given by
    \begin{align}
    	\label{eq:Un}
    	\cF U_n\cF^* (p) \coloneqq \sum_{k,\ell\in\cI} e^{in\left(\omega_\ell\left(p+\frac{\lambda}{n}\right) - \omega_k(p)\right)}   \Pi_k(p)   \Pi_\ell(p) = \sum_{k\in\cI} e^{in\left(\omega_k\left(p+\frac{\lambda}{n}\right) - \omega_k(p)\right)}   \Pi_k(p),
    \end{align}
    where the last equality is due to the fact that eigenspaces to different eigenvalues are orthogonal.
    
    By the triangle inequality, it clearly suffices to prove the estimates
    \begin{align}
    	\label{eq:proj}
        &\big\|W^{-n} e^{i\frac{\lambda}{n} X} W^n e^{-i\frac{\lambda}{n} X}
        - U_n\big\|
        \leq \frac{\lvert \lambda \rvert}{n}  
         \sum_{k\in\cI} \sup_{p\in\BT}\left\lVert   \Pi_k'(p) \right\rVert_{\cB(\mathcal{K})}
        \\&
        \label{eq:phase}
        	\left\lVert 
        	e^{i\lambda V}- U_n
        	\right\rVert\leq \frac{\lvert \lambda \rvert^2}{n} \sup_{k\in\cI,p\in\BT} \lvert \omega_k''(p) \rvert.
    \end{align}
	Note that for a local operator $T$ on $L^2(\BT;\cK)$, i.e., an operator acting as $Tf(p)=T(p)f(p)$ with some $T(p)\in\cB(\cK)$, is equivalent to estimating the operator norm pointwise, namely,
	\begin{align}
		\label{eq:supnorm}
		\|T\|_{\cB(L^2(\BT;\cK))} = \esssup_{p\in\BT}\|T(p)\|_{\cB(\cK)}.
	\end{align}
	By unitarity of the Fourier transform, we can thus estimate the defining expression pointwise in $\BT$.
 	More precisely, fixing,  $p\in\BT$, we find
    \begin{align*}
        &\big\lVert \mathcal{F}\big({W}^{-n} e^{i\frac{\lambda}{n} {X}}W^n e^{-i\frac{\lambda}{n} {X}} - U_n\big)\mathcal{F}^*(p)\big\rVert_{\cB(\cK)}=\left\lVert \sum_{k,\ell\in\cI} e^{in\left(\omega_\ell(p+\frac{\lambda}{n})-\omega_k(p)\right)}  \Pi _k(p)\big(  \Pi _\ell(p+\tfrac{\lambda}{n})-  \Pi_\ell(p)\big)\right\rVert_{\cB(\cK)}\\
        &\qquad\qquad\leq \left\lVert\sum_{k\in\cI} e^{-in\omega_k(p)}  \Pi _k(p)\right\rVert_{\cB(\cK)}\sum_{\ell\in\cI}\left\lVert e^{in\omega_\ell(p+\frac{\lambda}{n})}\big(  \Pi _\ell(p+\tfrac{\lambda}{n})-  \Pi _\ell(p)\big)\right\rVert_{\cB(\cK)}\\
        &\qquad\qquad\leq\frac{\lvert \lambda\rvert}{n} \sum_{\ell\in\cI}\sup_{p\in\BT}\left\lVert  \Pi_\ell'(p)\right\rVert_{\cB(\cK)},
    \end{align*}
	where
    in the last step, the first factor is one as the norm of a unitary operator and we used the vector-valued mean value theorem. In view of \cref{eq:supnorm}, this proves \cref{eq:proj}
    
    In order to show \cref{eq:phase}, we proceed similarly with the second expression in \cref{eq:Un} and note that again by the mean value theorem 
    \begin{align*}
    	\left\lVert 
    	\cF \big(e^{i\lambda V}- U_n\big)\cF(p)
    	\right\rVert_{\cB(\cK)}&\leq\left\lVert\sum_{k\in\cI}\left(e^{i\omega_k'(p)}-e^{in\left(\omega_k(p+\frac{\lambda}{n})-\omega_k(p)\right)}\right)  \Pi _k(p)\right\rVert_{\cB(\cK)}\\
        &\leq \sup_{k\in\cI} \left\lvert 1-e^{i\left(n\left(\omega_k(p+\frac{\lambda}{n})-\omega_k(p)\right)-\omega_k'(p)\right)}\right\rvert\\
        &\leq \sup_{k\in\cI} \left\lvert n\left(\omega_k\left(p+\frac{\lambda}{n}\right)-\omega_k(p)\right)-\omega_k'(p)\right\rvert\\
        &\leq\frac{\lvert\lambda\rvert^2}{n}\sup_{  k\in\cI,p\in\BT}\lvert\omega_k''(p)\rvert. \qedhere
    \end{align*}
\end{proof}
Thus we have completed the proof of \cref{Cor:triangle}.
Combining with \cref{lem:ZolotarevSim} gives us the
\begin{proof}[Proof of \cref{thm:main}]
    By \cref{lem:ZolotarevSim} we get for some universal $C>0$ and for any $\epsilon>0$
    \begin{align*}
        L(F_{Z_n}^\rho,F_V^\rho)\leq \epsilon+\epsilon^{-2}C\max\left\{\sup_{ 0<\lvert\lambda\rvert\leq 1}\left\lvert\frac{\widehat{F}_{Z_n}^\rho(\lambda)-\widehat{F}_V^\rho(\lambda)}{\lambda}\right\rvert+\sup_{ 1<\lvert\lambda\rvert}\left\lvert\frac{\widehat{F}_{Z_n}^\rho(\lambda)-\widehat{F}_V^\rho(\lambda)}{\lambda^2}\right\rvert\right\}
    \end{align*}
    We will concentrate on estimating the terms in the maximum separately by the same expression. For the difference of the characteristic functions appearing in both terms, we will use the estimate from \cref{Cor:triangle} for any $\lambda\in\IR$ and $n\in\IN$ 
    \begin{align*}
        \left\lvert\widehat{F}_{Z_n}^\rho(\lambda)-\widehat{F}_V^\rho(\lambda)\right\rvert\leq\frac{\lvert\lambda\rvert}{n}\left( \sup_{  k\in\cI,p\in\BT}\lvert \omega_k''(p)\rvert +\lvert\lambda\rvert\left(\tr(\lvert{X}\rvert \rho)+\sum_{k\in\cI}\sup_{p\in\BT}\lVert \Pi_k'(p)\rVert_{\mathcal{B}(\mathcal{H})}\right)\right)
    \end{align*}
    In the case $0< \lvert\lambda\rvert \leq 1$, this yields
    \begin{align*}
        \left\lvert\frac{\widehat{F}_{Z_n}^\rho(\lambda)-\widehat{F}_V^\rho(\lambda)}{\lambda}\right\rvert&\leq n^{-1}\left(\sup_{  k\in\cI,p\in\BT}\lvert \omega_k''(p)\rvert +\lvert\lambda\rvert\left(\tr(\lvert{X}\rvert \rho)+\sum_{k\in\cI}\sup_{p\in\BT}\lVert \Pi_k'(p)\rVert_{\mathcal{B}(\mathcal{H})}\right)\right)\\
        &\leq n^{-1}\left( \sup_{  k\in\cI,p\in\BT}\lvert \omega_k''(p)\rvert +\tr(\lvert{X}\rvert \rho)+\sum_{k\in\cI}\sup_{p\in\BT}\lVert \Pi_k'(p)\rVert_{\mathcal{B}(\mathcal{H})}\right).
    \end{align*}
    Further, for $\lvert\lambda\rvert\geq1$, we get
    \begin{align*}
        \left\lvert\frac{\widehat{F}_{Z_n}^\rho(\lambda)-\widehat{F}_V^\rho(\lambda)}{\lambda^2}\right\rvert&\leq n^{-1}\left(\lvert\lambda\rvert^{-1} \sup_{  k\in\cI,p\in\BT}\lvert \omega_k''(p)\rvert +\tr(\lvert{X}\rvert \rho)+\sum_{k\in\cI}\sup_{p\in\BT}\lVert \Pi_k'(p)\rVert_{\mathcal{B}(\mathcal{H})}\right)\\
        &\leq n^{-1}\left( \sup_{  k\in\cI,p\in\BT}\lvert \omega_k''(p)\rvert +\tr(\lvert{X}\rvert \rho)+\sum_{k\in\cI}\sup_{p\in\BT}\lVert\Pi_k'(p)\rVert_{\mathcal{B}(\mathcal{H})}\right).
    \end{align*}
    Combining these estimates, for any $\epsilon>0$ and $n\in\IN$, we have
    \begin{align*}
        L(\widehat{F}_{Z_n}^\rho-\widehat{F}_V^\rho)&\leq\epsilon +\epsilon^{-2}Cn^{-1}\left( \sup_{  k\in\cI,p\in\BT}\lvert \omega_k''(p)\rvert +\tr(\lvert{X}\rvert \rho)+\sum_{k\in\cI}\sup_{p\in\BT}\lVert \Pi_k'(p)\rVert_{\mathcal{B}(\mathcal{H})}\right).
    \end{align*}
    Choosing $\epsilon=n^{-1/3}$, we get the claim 
    \begin{align*}
        L(\widehat{F}_{Z_n}^\rho-\widehat{F}_V^\rho)&\leq n^{-\frac{1}{3}}C\left(C^{-1}+ \sup_{  k\in\cI,p\in\BT}\lvert \omega_k''(p)\rvert +\tr(\lvert{X}\rvert \rho)+\sum_{k\in\cI}\sup_{p\in\BT}\lVert \Pi_k'(p)\rVert_{\mathcal{B}(\mathcal{H})}\right).
        \qedhere
    \end{align*}
\end{proof}

\section{Berry-Esseen Bounds for Coin-Step Walks}
\label{sec:coinstep}
In this \lcnamecref{sec:coinstep}, we prove \cref{thm:quantum_rate_supnorm}.
We thus now restrict our attention to shift-coin walks with two degrees of freedom as introduced in \cref{subsec:Kolmo} with the general unitary coin matrix $C$ given by 
\begin{equation*}
   C\coloneqq e^{i\theta}\begin{pmatrix}
        a&b\\
        -\bar{b}&\bar{a}
    \end{pmatrix}
\end{equation*} 
for $a,b\in\C$ with $\lvert a\rvert^2+\lvert b\rvert^2=1$ and $\theta\in\R$.
Furthermore, w.l.o.g., we will restrict our attention to pure states of the form $\rho=\ket{\delta_0\phi}\bra{\delta_0\phi}$ for some $\phi\in\IC^2$
and denote the corresponding CDFs by $F^\phi_A$, similar to \cref{eq:CDF}.

The analysis of the distances between the CDFs of the rescaled finite time distribution and the limiting distribution can be separated into two parts which will correspond to the estimates in two distinct regions in position space: (1) the inner part of the propagation region, where the limiting CDF is smooth and we can thus apply a simple result relating supremum and L\'evy metric, and (2) the boundary of the propagation region, where we need to refer to
a stationary phase analysis of the wavefront amplitudes, building on previous results from \cite{SunadaTate.2012}.

Throughout this section, we will use the usual $O$-notation, i.e., $f(n)=O(g(n))$ as $n\to\infty$, if there exists a (universal in the parameters of $f$) $C>0$ such that $|f(n)|\le C|g(n)|$ for all $n$ sufficiently large. 

\subsection{The Asymptotic Distribution}
Let us start, by recalling Konno's result on the asymptotic distribution.
\begin{prop}[{\cite[Thm.~1]{Konno.2005b}}]
	\label{lem:asymptotic}
	Assume $a\ne 0$ and, for $\phi\in\IC^2$, define
	\begin{align*}
		\lambda_C(\phi) \coloneqq |\phi_2|^2-|\phi_1|^2 + \frac{1}{|a|^2}\big({\bar{a}}b \bar{\phi_1}\phi_2 + {a}\bar b \phi_1 \bar{\phi_2}\big).
	\end{align*}
	Then the asymptotic CDF $F_V^\phi$ has the density
	\begin{align*}
		\sigma_{C,\phi} (x) \coloneqq \begin{cases} \dfrac{|b|(1+\lambda_C(\phi)x)}{\pi(1-x^2)\sqrt{|a|^2-x^2}} & \mbox{if}\ |x|<|a|, \\ 0 & \mbox{else}.
								\end{cases}
	\end{align*}
\end{prop}

From the explicit density, we deduce growth bounds for the limiting distribution function in the following two lemmas that we apply later for the error estimates close to the wavefront. 

\begin{lem}\label{lem:densityestimate}
    For $\epsilon>0$, we have
    \begin{align*}
        \sigma_{C,\phi}(\mp\lvert a\rvert \pm\epsilon)=\epsilon^{-\frac{1}{2}}\frac{\lvert b\rvert (1\mp\lambda_C(\phi)\lvert a\rvert)}{\pi (1-\lvert a\rvert^2)\sqrt{2\lvert a\rvert}}+ O(\epsilon^{\frac{1}{2}}), \qquad \epsilon\to0. 
    \end{align*}
\end{lem}
\begin{proof}
    We start with the definition of the density and rewrite
    \begin{align*}
        \sigma_{C,\phi}(-\lvert a\rvert +\epsilon)&=\frac{\lvert b\rvert (1+\lambda_C(\phi)(-\lvert a\rvert +\epsilon))}{\pi(1-(-\lvert a\rvert +\epsilon)^2)\sqrt{\lvert a\rvert^2-(-\lvert a\rvert +\epsilon)^2}}\\
        &=\frac{\lvert b\rvert(1-\lambda_C(\phi)(\lvert a\rvert-\epsilon))}{\pi(1-\epsilon^2-\lvert a\rvert^2+2\lvert a\rvert \epsilon)\sqrt{2\lvert a\rvert \epsilon-\epsilon^2}}\\
        &=\epsilon^{-\frac{1}{2}}\frac{\lvert b\rvert (1-\lambda_C(\phi)\lvert a\rvert)}{\pi (1-\lvert a\rvert^2)\sqrt{2\lvert a\rvert}}\left(1+\epsilon\frac{2\lvert a\rvert -\epsilon}{1-\lvert a\rvert^2}\right)^{-1}\left(1-\frac{\epsilon}{2\lvert a\rvert}\right)^{-\frac{1}{2}}\left(1+\frac{\lambda_C(\phi)\epsilon}{1-\lambda_C(\phi)\lvert a\rvert}\right)\\
        &=\epsilon^{-\frac{1}{2}}\frac{\lvert b\rvert (1-\lambda_C(\phi)\lvert a\rvert)}{\pi (1-\lvert a\rvert^2)\sqrt{2\lvert a\rvert}}+ O(\epsilon^{\frac{1}{2}}).\qedhere
    \end{align*}
\end{proof}
Integrating yields
\begin{cor}
	\label{cor:asupbound}
	For all $\eps\ge 0$, we have
    \begin{align*}
        F_{V}^{\phi}(-\lvert a\rvert+ n^{-\frac{2}{3}-\eps}) = O(n^{-\frac{1}{3}-\frac\eps2}), \qquad F_V^\phi(|a|-n^{-\frac23-\eps})= O(n^{-\frac{1}{3}-\frac\eps2}), \qquad n\to\infty.
    \end{align*}
\end{cor}
\begin{proof}
	We only prove the first estimate here. By symmetry, the second follows similarly.
	
    Combining \cref{lem:asymptotic,lem:densityestimate}, there is $C>0$ such that
    \begin{align*}
        F_{V}^{\phi}(-|a|+n^{-\frac{2}{3}-\eps})&\leq\int_{0}^{n^{-\frac{2}{3}-\eps}} C (x^{-\frac{1}{2}}+x^\frac{1}{2})\d x\\
        &{=}C\left({2}n^{-\frac{1}{2}\frac{2}{3}-\frac\eps2}{-\tfrac{2}{3}}n^{-\frac{3}{2}(\frac{2}{3}+\eps)}\right)\\
        &= O(n^{-\frac{1}{3}-\frac\eps2}). \qedhere
    \end{align*}
\end{proof}

\subsection[Smooth Region]{Smooth Region $|x|<|a|$}
We first show that in the smooth part the order of convergence of the Levy metric can be directly translated into a bound on the sup-distance between the distributions.
\begin{lem}
	\label{lem:smooth}
	Let $F$ and $G$ be CDFs and let $G$ be continuously differentiable on an open interval $(\alpha,\beta)$.
	Then for any compact interval $I=[\alpha',\beta']\subset (\alpha,\beta)$, if $L(F,G)<\min\{\alpha'-\alpha,\beta-\beta'\}$, we have
	\begin{align*}
		\|(F-G)\chr_I\|_\infty \le \big(1+\|G'\chr_I\|_\infty\big)L(F,G).
	\end{align*}
\end{lem}
\begin{proof}
	We abbreviate $L\coloneqq L(F,G)$ in this proof. By definition \cref{eq:Levy}, for any $x\in\IR$, $\eps>L$, we have
	\[  G(x-\eps)-\eps\le F(x)\le G(x+\eps)+\eps. \]
	For $x\in I$, be continuity, we can take the limit $\eps\to L$ and obtain
	\begin{align*}
		G(x-L)-G(x)-L \le F(x)-G(x) \le G(x+L) - G(x)+L.
	\end{align*}
	The assumptions further yield $x\pm L\in(\alpha,\beta)$, whence we can apply the mean value theorem and find
	\begin{align*}
		\big|F(x)-G(x)\big| \le \big|G(x\pm L) - G(x)\Big| + L \le L\left(\sup_{x'\in (x\pm L)}|G'(x)| + 1\right).
	\end{align*}
	This immediately yields the statement.
\end{proof}
Combining \cref{thm:main,lem:asymptotic} thus yields:
\begin{cor}\label{lem:inside}
    For any $r>0$ and any $\phi\in\HS$, there exists $C_r>0$ such that for all $n\in\IN$
    \begin{align*}
	        \sup_{\lvert x\rvert \leq \lvert a\rvert -r}\left\lvert F_{X_n}^{\phi}(x)-F_V^{\phi}(x)\right\rvert\leq C_r n^{-\frac 1 3 }.
    \end{align*}
\end{cor}

\subsection[Wavefront Region]{Wavefront Region $|x|\approx |a|$}\label{sec:wavefront}
For the wavefront region, our results heavily rely on the explicit expression for the density of $F_{X_n}^\phi$, calculated by Sunada and Tate in \cite{SunadaTate.2012}.
In the following, we recall and extend some of their results. We then apply them to prove lower and upper bounds on $F_{X_n}^\phi$ in close proximity to $\pm|a|$. We remark that, further outside the propagation region, exponential decay of occupation probability has been proven \cite{SunadaTate.2012,CedzichJoyeWerner2.2025}.

We define the density and CDF after $n$ steps  by
\begin{align}
	\label{eq:pn}
	p_n(\phi;k) \coloneqq \tr\big(\rho\chr_{\{k\}}W^{-n}X W^n\big), \qquad F_n(\phi;x) \coloneqq  \sum_{k\le x}p_n(\phi;k),
\end{align}
thus yielding
\begin{align}
	\label{eq:FXnFn}
	F_{X_n}^\phi(x) = F_n(\phi;nx).
\end{align}
Furthermore, we denote the Airy function by
\begin{align*}
	\operatorname{Ai}(x) \coloneqq \frac{1}{2\pi}\int_{-\infty}^\infty e^{i\xi^3/3 + i\xi x}\d\xi = \frac1\pi\int_0^\infty \cos(\frac13\xi^3 + \xi x)\d\xi,
\end{align*}
see for example \cite[\S\,2.8]{Olver.1974}.
\begin{prop}\label{thm:tate}
	Assume $a\ne0\ne b$ and $\phi\in\IC^2$.
    \begin{enumprop}
        \item \label{thm:tate.1}
        For a sequence of integers $(y_n)_{n\in\N}\subset\IZ$ with $y_n=\pm n\lvert a\rvert + d_n$ such that $d_n= O(n^{\frac{1}{3}})$, $n\to\infty$, the transition probability is given by
        \begin{equation*}
            p_n(\phi;y_n)=(1+(-1)^{n+y_n})\alpha^2n^{-\frac{2}{3}}\left\lvert \textup{Ai}(\pm\alpha n^{-\frac{1}{3}}d_n)\right\rvert^2(1\pm\lvert a\rvert\lambda_A(\phi))+O(n^{-1}),\quad n\to\infty
        \end{equation*}
        with $\alpha=(2/\lvert a\rvert\lvert b\rvert ^2)^{\frac{1}{3}}$.
        \item \label{thm:tate.2} There exist $r>0$ and functions $p_\pm,q_\pm,s_\pm\in C^\infty(\IR)$ depending on $\phi$ such that for any $n\in\IN$ and $k\in\IZ$ with $\lvert k\rvert \leq rn$, we uniformly in $k\in\IZ$ have
        \begin{align*}
            &p_n\left(\phi,\pm\lfloor n\lvert a\rvert\rfloor \mp k\right) \\&\qquad= (1+(-1)^k)\left(n^{-\frac{2}{3}}s_\pm^2\left(\tfrac{k}{n}\right)\textup{Ai}^2\big(\pm n^{\frac{2}{3}}p_\pm\left(\tfrac{k}{n}\right)\big)+n^{-\frac{4}{3}}q_\pm^2\left(\tfrac{k}{n}\right)\textup{Ai}'^2\big(\pm n^{\frac{2}{3}}p_\pm\left(\tfrac{k}{n}\right)\big)\right)\\&\qquad\qquad+O(n^{-\frac{4}{3}}), \qquad n\to\infty,
        \end{align*}
        where $p_\pm(0)=0$ and $p_\pm'(0)=\alpha>0$. 
    \end{enumprop}
\end{prop}   
\begin{proof}
	\subcref{thm:tate.1} is {\cite[Thm.~1.3]{SunadaTate.2012}}.
	\subcref{thm:tate.2}
     follows from a careful analysis of the terms in the proof of the stationary phase argument in \cite{SunadaTate.2012}, see \cref{app:stationary}.
\end{proof}
For the remainder of this \lcnamecref{sec:wavefront}, we assume $a\ne 0 \ne b$ and fix $\phi\in\IC^2$.
We now integrate the terms appearing in the above asymptotic analysis.
The first result proves the lower bound.
\begin{lem}\label{lem:suplowbound}
    There exists $C>0$ such that for any $n\in\IN$
    \begin{align*}
        \min\big\{F_n(\phi;-n\lvert a\rvert), 1 - F_n(\phi;n|a|)\big\}\geq C n^{ -\frac{1}{3}}.
    \end{align*}
\end{lem}
\begin{proof}
    We will restrict our attention to $F_n(\phi;-n|a|)$, because the statement for $F_n(\phi;n|a|)$ then directly follows for symmetry reasons, and show the claim, by summing the transition probabilities of \cref{thm:tate.1}. As $\operatorname{Ai}(0)>0$ \cite[\S\,2,~(8.04)]{Olver.1974}, the Airy function is strictly positive in an open interval around zero and thus there exist $C_0>0$ and $C_1>0$ such that for any $0\leq k\leq C_0 n^{\frac{1}{3}}$ 
    \begin{align*}
        \textup{Ai}(-\alpha n^{-\frac{1}{3}}k)\geq C_1.
    \end{align*}
	Thus inserting the estimate \cref{thm:tate.1} into \cref{eq:pn} yields
    \begin{align*}
        F_{n}(\phi,-n\lvert a\rvert)&=\sum_{k=0}^\infty p_n(\phi;-k-\lfloor n\lvert a\rvert\rfloor)\\
        &\geq \sum_{k=0}^{ n^\frac{1}{3}} (1+(-1)^{n+k})\alpha^2n^{-\frac{2}{3}}\left\lvert \textup{Ai}(-\alpha n^{-\frac{1}{3}}d_n)\right\rvert^2(1-\lvert a\rvert\lambda_A(\phi))+O(n^{-1})\\
        &\geq \frac 12C_1^2\alpha^2(1-\lvert a\rvert \lambda_A(\phi)) n^{-\frac{2}{3}}  {n^{\frac{1}{3}}}+O(n^{-\frac{2}{3}})\\
        &\geq C_2\alpha^2(1-\lvert a\rvert \lambda_A(\phi)) n^{-\frac{1}{3}}
    \end{align*}
    where we chose $C_2>0$ such that the higher order error gets absorbed.
\end{proof}
This already provides us with the claimed lower bound for the sup-norm, since $F_V^\phi(-|a|)=0$.
It also allows to prove the lower bound on the L\'evy metric.
\begin{rem}[Proof of \cref{rem:lowerboundLevy}]\label{rem:lowerboundLevyproof}
	The upper bound was already proven in \cref{thm:main}.
	For the lower bound note that combining \cref{lem:suplowbound,cor:asupbound}, for all $\eps>0$, we have
	\[ F_V^\phi(-|a|+n^{-\frac 2 3 -\eps})+n^{-\frac 2 3 -\eps}\lesssim n^{-\frac 1 3 -\eps/2} + n^{-\frac 2 3 -\eps} \lesssim n^{-\frac 1 3 } \lesssim F_{X_n}^\phi(-|a|), \]
	which proves $n^{-2/3-\eps}\lesssim L(F_V^\phi,F_{X_n}^\phi)$.
	\qed
\end{rem}
Thus let us now come to the upper bound on the sup-norm, where we will first analyze the CDF after $n$ steps.
\begin{lem}\label{thm:leftbound}
	We have
    \begin{align*}
        F_n(\phi;-n\lvert a \rvert +n^{\frac{1}{3}})= O(n^{-\frac{1}{3}}), \quad 1-F_n(\phi;n|a|-n^{\frac13}) = O(n^{-\frac{1}{3}}), \quad n\to\infty.
    \end{align*}
\end{lem}
\begin{proof}
	We again restrict our attention to the first estimate, since $1-F_n(\phi;k)=F_n(C^*\sigma_x C\phi;-k)$.
	
    For the purpose of this proof, we consider $F_n(\phi;\cdot)$ on the three regions $(-\infty,-Cn^{\frac{2}{3}}], [-Cn^{\frac{2}{3}},-n^{\frac{1}{3}}],[-n^{\frac{1}{3}},n^{\frac{1}{3}}]$, where we pick $C>0$ such that $L(F_{X_n}^\phi,F_V^\phi)< Cn^{-1/3} $ for all $n\in\IN$. This is possible by \cref{thm:main}.
    
    For the first region, we directly apply the bound on the L\'evy metric \cref{thm:main} and the fact that $F_{V}^{\phi}(k)=0$ for $k\le-|a|$, cf. \cref{lem:asymptotic}. With the choice $L_n \coloneqq \frac12 L(F_{X_n}^\phi,F_V^\phi) + \frac12 Cn^{-1/3}=O(n^{-1/3})$, this yields
    \begin{align*}
        F_n(\phi;-n\lvert a\rvert -Cn^{\frac{2}{3}}) = F_{X_n}^\phi(-|a|-Cn^{-\frac13})\leq F_{V}^{\phi}(-\lvert a\rvert - Cn^{-\frac{1}{3}}+L_n)+L_n \leq 0+ Cn^{-\frac{1}{3}}.
    \end{align*} 
    For the second region, we apply the following estimates \cite[\S\,11.1.4]{Olver.1974}: setting $f_\pm(x)=e^{-\frac{4}{3}x^{\frac{3}{2}}}x^{\pm\frac{1}{2}}$ and given $x_0>0$, there is $C_{x_0}\ge 1$ such that
    \begin{align}
    	\label{eq:Ai2est}
        \textnormal{Ai}^2(x)\leq f_-(x),\qquad
        \lvert\textnormal{Ai}'^2(x)\rvert\leq  C_{x_0}f_+(x), \qquad x\ge x_0.
    \end{align}
    Now let the function $p$ be $p_-$ from \cref{thm:tate.2} and recall that 
    $p(0)=0$ and $p'(0)=\alpha>0$.
    Especially, for $n$ sufficiently large $p'(x)\ge \frac\alpha2$ for all $x\in[0,n^{-1/3}]$.
    Thus, since $f_\pm$ are positive and decreasing, we find
    \begin{align}
    	\label{eq:fpmest}
        \begin{aligned}
        	\sum_{k=\lfloor n^{\frac{1}{3}}\rfloor}^{\lfloor n^{\frac{2}{3}}\rfloor}  f_\pm(n^{\frac{2}{3}}p(\tfrac{k}{n}))&\leq \int_{0}^{n^{\frac{2}{3}}} f_\pm(n^{\frac{2}{3}}p(\tfrac{x}{n}))\d x
        	\leq \int_{0}^{n^{\frac{2}{3}}p(n^{-\frac{1}{3}})}\frac{n^{\frac{1}{3}}}{p'(p^{-1}(un^{-\frac{2}{3}}))} f_\pm(u)\d u\\
        	&\leq \frac{2}{\alpha} n^{\frac{1}{3}}\int_0^{n^{\frac{2}{3}}p(n^{-\frac{1}{3}})} f_\pm(u)\d u
        	\leq \frac{2}{\alpha} n^{\frac{1}{3}}\int_0^{\infty} f_\pm(u)\d u \\& = O (n^{-\frac 1 3 }), \qquad n\to\infty.
        \end{aligned}
    \end{align}
    Further, observing that 
     $n^{2/3}p(n^{-2/3})\xrightarrow{n\to\infty}\alpha$, and thus especially $n^{2/3}p(n^{-2/3})\geq \alpha/2$ for $n$ sufficiently large,
    we can insert \cref{eq:fpmest} into \cref{eq:Ai2est} and obtain
    \begin{align*}
        \sum_{k=\lfloor n^{\frac{1}{3}}\rfloor}^{\lfloor n^{\frac{2}{3}}\rfloor} \textup{Ai}^2(-n^{\frac{2}{3}}p(-\tfrac{k}{n}))&\leq \sum_{k=\lfloor n^{\frac{1}{3}}\rfloor}^{\lfloor n^{\frac{2}{3}}\rfloor}  f_-(n^{\frac{2}{3}}p(\tfrac{k}{n}))=O(n^{\frac{1}{3}}), \quad n\to\infty,\\
        \sum_{k=\lfloor n^{\frac{1}{3}}\rfloor}^{\lfloor n^{\frac{2}{3}}\rfloor} \textup{Ai}'^2(-n^{\frac{2}{3}}p(-\tfrac{k}{n}))&\leq C_{\alpha/2}\sum_{k=\lfloor n^{\frac{1}{3}}\rfloor}^{\lfloor n^{\frac{2}{3}}\rfloor}  f_+(n^{\frac{2}{3}}p(\tfrac{k}{n}))=O(n^{\frac{1}{3}}), \quad n\to\infty.
    \end{align*}
    Since the functions $q=q_-$ and $s=s_-$ from \cref{thm:tate.2} are smooth, they are uniformly bounded on the domain under consideration, say by a constant $C>0$.
    Thus, we can estimate further
    \begin{align*}
        \sum_{k=\lfloor n^{\frac{1}{3}}\rfloor}^{\lfloor n^{\frac{2}{3}}\rfloor} p_n(\phi,-\lfloor n\lvert a\rvert\rfloor +k)&\leq 2C(n^{-\frac{2}{3}}Cn^{\frac{1}{3}}-n^{-\frac{4}{3}}Cn^{\frac{1}{3}}) = O(n^{-\frac{1}{3}}), \quad n\to\infty.
    \end{align*}
    In the last region $[-n^{\frac13},n^{\frac13}]$, we can directly apply the result of \cref{thm:tate.1} with the additional observation that the Airy function term is uniformly bounded on the interval of consideration. This leads to the estimate 
    \begin{align*}
        \sum_{k=-\lfloor n^{\frac{1}{3}}\rfloor}^{\lfloor n^{\frac{1}{3}}\rfloor} p_n(\phi,-\lfloor n\lvert a\rvert\rfloor +k)&\leq \sum_{k=-\lfloor n^{\frac{1}{3}}\rfloor}^{\lfloor n^{\frac{1}{3}}\rfloor} C n^{-\frac{2}{3}} + O(n^{-1})\\
        &\leq 2 n^{\frac{1}{3}}(C n^{-\frac{2}{3}}+O(n^{-1}))\\
        & =  O(n^{-\frac{1}{3}}).
    \end{align*}
    This completes the proof. 
 \end{proof}
By extracting the leading terms from the Airy function, we now show that the error in the transition probabilities is given by an oscillating function. 
\begin{prop}\label{thm:decomposition}
    There exists $r>0$ such that
    \begin{align*}
        \textup{OSC}^\pm_n\left(\tfrac{k}{n}\right)\coloneqq p_n\left(\phi,\pm\lfloor n\lvert a\rvert\rfloor \mp k\right) - \frac 1n \sigma_{C,\phi}\left({\pm \lvert a\rvert\mp}\tfrac{k}{n}\right)
    \end{align*}
    satisfies
    \begin{align}
    	\label{eq:OSCest}
        \sum_{k=\lfloor n^{\frac{1}{3}}\rfloor}^{\lfloor rn\rfloor}\textup{OSC}_n^\pm\left(\tfrac{k}{n}\right)= O(n^{-\frac{1}{3}}).
    \end{align}
\end{prop}
\begin{proof}
	We only prove the result in the region close to $-|a|$. By symmetry, the second statement follows.
	
    We can rewrite the Airy function components by using \cite[\S\,4.4]{Olver.1974}
    \begin{align*}
        \textup{Ai}(-z)&=\pi^{-\tfrac{1}{2}}z^{-\tfrac{1}{4}}\left(\cos\left(\tfrac{2}{3}z^{\frac{3}{2}}-\tfrac{\pi}{4}\right)+O(z^{-3})\right),\\
        \textup{Ai}'(-z)&= \pi^{-\tfrac{1}{2}}z^{\tfrac{1}{4}}\left(\sin\left(\tfrac{2}{3}z^{\frac{3}{2}}-\tfrac{\pi}{4}\right)+O(z^{-3})\right).
    \end{align*}
    We thus have
    \begin{align*}
        \textup{Ai}^2(-z)&=\pi^{-1}z^{-\tfrac{1}{2}}\cos^2\left(\tfrac{2}{3}z^{\frac{3}{2}}-\tfrac{\pi}{4}\right)+O(z^{-\frac{13}{4}})=\pi^{-1}z^{-\tfrac{1}{2}}\left(1+\cos\left(\tfrac{4}{3}z^{\frac{3}{2}}-\tfrac{\pi}{2}\right)\right)+O(z^{-\frac{13}{4}}),\\
        &=\pi^{-1}z^{-\tfrac{1}{2}}\left(1+\sin\left(\tfrac{4}{3}z^{\frac{3}{2}}\right)\right)+O(z^{-\frac{13}{4}})\\
        \textup{Ai}'^2(-z)&=\pi^{-1}z^{\tfrac{1}{2}}\sin^2\left(\tfrac{2}{3}z^{\frac{3}{2}}-\tfrac{\pi}{4}\right)+O(z^{-\frac{11}{4}})=\pi^{-1}z^{\tfrac{1}{2}}\left(1-\cos\left(\tfrac{4}{3}z^{\frac{3}{2}}-\tfrac{\pi}{2}\right)\right)+O(z^{-\frac{11}{4}})\\
        &=\pi^{-1}z^{\tfrac{1}{2}}\left(1-\sin\left(\tfrac{4}{3}z^{\frac{3}{2}}\right)\right)+O(z^{-\frac{11}{4}}).
    \end{align*}
    Setting $z=n^{\frac{2}{3}}p\left(\tfrac{k}{n}\right)$ and combining with \cref{thm:tate.2}, where we again set $p=p_-$, $q=q_-$, $s=s_-$, we then get for the decomposition of $p_n$
    \begin{align*}
        &p_n(\phi,-\lfloor n\lvert a\rvert\rfloor +k)\\&= (1+(-1)^k)\left(n^{-\frac{2}{3}}s^2\left(\tfrac{k}{n}\right)\textup{Ai}^2\left(-n^{\frac{2}{3}}p\left(\tfrac{k}{n}\right)\right)+n^{-\frac{4}{3}}q^2\left(\tfrac{k}{n}\right)\textup{Ai}'^2\left(-n^{\frac{2}{3}}p\left(\tfrac{k}{n}\right)\right)\right)+O(n^{-\frac{4}{3}})\\
        &=n^{-1}\left(\pi^{-1}p^{-\frac{1}{2}}\left(\tfrac kn\right)s^2(\xi)+\pi^{-1}p^{\frac{1}{2}}\left(\tfrac kn\right)q^2\left(\tfrac kn\right)\right)\\&\phantom{=}+\pi^{-1}n^{-1}\sin\left(\tfrac{4}{3}np^{\frac{3}{2}}\left(\tfrac kn\right)\right)\left(p^{-\frac{1}{2}}\left(\tfrac kn\right)s^2\left(\tfrac kn\right)-p^{\frac{1}{2}}\left(\tfrac kn\right)q^2\left(\tfrac kn\right)\right) +O(n^{-\frac{4}{3}})\\
        &=n^{-1 }f\left(-\lvert a\rvert+\tfrac{k}{n}\right)+\textup{OSC}_n\left(\tfrac{k}{n}\right).
    \end{align*}
    with the choices
    \begin{align*}
        f(-\lvert a\rvert+\xi)&= \pi^{-1}p^{-\frac{1}{2}}(\xi)s^2(\xi)+\pi^{-1}p^{\frac{1}{2}}(\xi)q^2(\xi)\\
        \textup{OSC}_n(\xi)&=\pi^{-1}n^{-1}\sin\left(\tfrac{4}{3}np^{\frac{3}{2}}(\xi)\right)\left(p^{-\frac{1}{2}}(\xi)s^2(\xi)-p^{\frac{1}{2}}(\xi)s^2(\xi)\right) +O(n^{-\frac{4}{3}}).
    \end{align*}
    We now prove $f=\sigma_{C,\phi}$ on $(0,r)$ given \cref{eq:OSCest}. To this end, for any $r'\in[0,r]$, we have
    \begin{align*}
    	&\left\lvert \int_{n^{-\frac{2}{3}}}^{r'}f(-\lvert a\rvert+x)-\sigma_{C,\phi}(-\lvert a\rvert+x)\d x\right\rvert\\
    	&\qquad\leq \left\lvert \sum_{k=n^\frac{1}{3}}^{r'n}f\left(-\lvert a\rvert+\tfrac{k}{n}\right)-\int_{n^{-\frac{2}{3}}}^{r'}f(-\lvert a\rvert+x)\d x\right\rvert+\left\lvert  \sum_{k=n^\frac{1}{3}}^{r'n}\textup{OSC}_n\left(\tfrac{k}{n}\right)\right\rvert \\&\qquad\qquad +\left\lvert \sum_{k=n^\frac{1}{3}}^{r'n}p_n\left(-\lvert a\rvert+\tfrac{k}{n}\right)-\int_{n^{-\frac{2}{3}}}^{r'}\sigma_{C,\phi}(x)\d x \right\rvert\\
    	&\qquad\leq O(n^{-\frac{1}{3}}) + \left\lvert F_{n}(\phi;-n\lvert a\rvert + r'n)-F_{V}^{\phi} (-\lvert a\rvert + r')\right\rvert + \left\lvert F_{n}(\phi;-n\lvert a\rvert + n^\frac{1}{3})-F_{V}^{\phi} (-\lvert a\rvert + n^{-\frac{2}{3}})\right\rvert\\
    	&\qquad\leq  \left\lvert F_{X_n}^{\phi}(\lvert a\rvert + r')-F_{V}^{\phi} (-\lvert a\rvert + r')\right\rvert + O(n^{-\frac{1}{3}})\xrightarrow{n\to\infty} 0,
    \end{align*}
    by the Portmonteau theorem. Since this holds for arbitrary $r'\in[0,r]$ and since both $f$ and $\sigma_{C,\phi}$ are continuous, we have $f=\sigma_{C,\phi}$.
    
    It thus remains to prove \cref{eq:OSCest}. We will deduce this claim from estimates on oscillatory sums, collected in \cref{app:osc}.
    We therefore split
    \begin{align}
    	\label{eq:OSCsplit}
    	\left|\sum_{k=\lfloor n^{\frac{1}{3}}\rfloor}^{\lfloor rn\rfloor}\textup{OSC}_n\left(\tfrac{k}{n}\right)\right| \le \left|\sum_{k=\lfloor n^{\frac{1}{3}}\rfloor}^{\lfloor n^{\frac{2}{3}}\rfloor}\textup{OSC}_n\left(\tfrac{k}{n}\right)\right| + \left|\sum_{k=\lfloor n^{\frac{2}{3}}\rfloor}^{\lfloor rn\rfloor}\textup{OSC}_n\left(\tfrac{k}{n}\right)\right|.
    \end{align}
    First, by continuity of $s$, \cref{lem:OSCest2} implies
    \begin{align}
    	\label{eq:OSC1}
    	\sum_{k=\lfloor n^{\frac 1 3 }\rfloor}^{\lfloor n^{\frac 2 3 }\rfloor}n^{-1}\sin\left(\tfrac{4 }{3}n p\left(\tfrac{k}{n}\right)^{\frac{3}{2}}\right)s^2\left(\tfrac{k}{n}\right)p\left(\tfrac{k}{n}\right)^{-\frac{1}{2}}& = O\left(n^{-\frac{1}{3}}\right)
    \end{align}
    Further, applying \cref{cor:cancel} with $s_n=\lfloor n^{2/3}\rfloor$ and $f(x)=s^2(x)p^{-1/2}(x)$ for $r$ sufficiently small that $f$ is monotonically decreasing and positive on $(0,r)$,
    we find $C>0$
    \begin{align*}
    	\left|\sum_{k=\lfloor n^{\frac{2}{3}}\rfloor}^{\lfloor rn\rfloor}n^{-1}\sin\left(\tfrac{4 }{3}np\left(\tfrac{k}{n}\right)^{\frac{3}{2}}\right)s^2\left(\tfrac{k}{n}\right)p\left(\tfrac{k}{n}\right)^{-\frac{1}{2}}\right|
    	\le
    	Cn^{-\frac 1 2 }s^2(r)p^{-\frac 1 2 }(r) + Cn^{-\frac 1 2 }s^2(n^{-\frac 1 3 })p^{-\frac 1 2 }(n^{-\frac 1 3 }).
    \end{align*}
    Since $p(0)=0$ and $p'(0)>0$, we find $p^{-1/2}(n^{-1/3})=O(n^{1/6})$, so by continuity of $s$ and $p$, we conclude
    \begin{align}
    	\label{eq:OSC2}
    	\sum_{k=\lfloor n^{\frac{2}{3}}\rfloor}^{\lfloor rn\rfloor}n^{-1}\sin\left(\tfrac{4 }{3}np\left(\tfrac{k}{n}\right)^{\frac{3}{2}}\right)s^2\left(\tfrac{k}{n}\right)p\left(\tfrac{k}{n}\right)^{-\frac{1}{2}}
    	= O(n^{-\frac 1 3 })
    \end{align}
    Similarly, by applying \cref{cor:cancel} with $s_n=n^{1/3}$ and $f(x)=q^2(x)p^{1/2}(x)$,  we find
    \begin{align}
    	\label{eq:OSC3}
    	\begin{aligned}
    		&\left|\sum_{k=\lfloor n^{\frac{1}{3}}\rfloor}^{\lfloor rn\rfloor}n^{-1}\sin\left(\tfrac{4 }{3}np\left(\tfrac{k}{n}\right)^{\frac{3}{2}}\right)q^2\left(\tfrac{k}{n}\right)p\left(\tfrac{k}{n}\right)^{\frac{1}{2}}\right|\\
    		&\qquad \le Cn^{-\frac 1 2 }q^2(r)p^{}(r) + Cn^{-\frac 1 2 }q^2(n^{-\frac 2 3 })p^{\frac 1 2 }(n^{-\frac 2 3 }) = O(n^{-\frac 1 2 }).
    	\end{aligned}
    \end{align}
    Inserting \cref{eq:OSC1,eq:OSC2,eq:OSC3} into \cref{eq:OSCsplit} proves \cref{eq:OSCest} and thus finishes the proof.
\end{proof}
We can conclude with our main result on the sup-distance between the CDFs in the area around the wavefront.
\begin{prop}\label{thm:wallerror}
    There exists $r>0$ and $C>0$ such that for all $r'\leq r$ 
    \begin{align*}
        \left\lvert F_{n}(\phi;\pm n\lvert a\rvert \mp nr')-F_{V}^{\phi}(\pm\lvert a\rvert \mp r')\right\rvert\leq Cn^{-\frac{1}{3}}.
    \end{align*}
\end{prop}
\begin{proof}
    For $r'\leq 0 $ this follows from monotonicity of the distribution functions and \cref{thm:leftbound}, so from now we only consider $r'>0$. 
    Then, applying \cref{thm:decomposition,thm:main}, we find
    \begin{align*}
        &\left\lvert F_{S_n}^{\phi}(-n\lvert a\rvert + r'n)-F_{V}^{\phi} (-\lvert a\rvert + r')\right\rvert\\
        &\leq \left\lvert \sum_{k=n^\frac{1}{3}}^{r'n}p_n\left(-\lvert a\rvert+\frac{k}{n}\right)-\int_{n^{-\frac{2}{3}}}^{r'}\sigma_{C,\phi}(-\lvert a\rvert+x)\d x \right\rvert + \left\lvert F_{S_n}^{\phi}(-n\lvert a\rvert + n^\frac{1}{3})-F_{V}^{\phi} (-\lvert a\rvert + n^{-\frac{2}{3}})\right\rvert\\
        &\leq \left\lvert \sum_{k=n^\frac{1}{3}}^{r'n}\sigma_{C,\phi}\left(-\lvert a\rvert+\frac{k}{n}\right)-\int_{n^{-\frac{2}{3}}}^{r'}\sigma_{C,\phi}(-\lvert a\rvert+x)\d x\right\rvert+\left\lvert  \sum_{k=n^\frac{1}{3}}^{r'n}\textup{OSC}_n\left(\frac{k}{n}\right)\right\rvert+O(n^{-\frac{1}{3}})\\
        & = \left\lvert \sum_{k=n^\frac{1}{3}}^{r'n}\sigma_{C,\phi}\left(-\lvert a\rvert+\frac{k}{n}\right)-\int_{n^{-\frac{2}{3}}}^{r'}\sigma_{C,\phi}(-\lvert a\rvert+x)\d x\right\rvert + O(n^{-\frac13}).
    \end{align*}
    It thus remains to estimate the first term, by proving that there exists $C>0$ such that for any $0<r'\leq r$
    \begin{align*}
    	\left\lvert \sum_{k=n^{\frac{1}{3}}}^{r'n} \sigma_{C,\phi}\left(\tfrac{k}{n}\right)n^{-1}-\int_{n^{-\frac{2}{3}}}^{r'}\sigma_{C,\phi}(x)\d x\right\rvert\leq C n^{-\frac{1}{2}} .
    \end{align*}
    To see this, observe that $\sigma_{C,\phi}$ is Lipschitz continuous on the interval $I_n\coloneq [n^{-\frac{1}{3}},r']$ with Lipschitz constant $L_n=Cn^{\frac{1}{2}}$ where $C>0$ is a constant independent of $n$.
    Then, treating the sum as a Riemann sum and using the approximation error formula with 
    \begin{align*}
    	\lvert I_n\rvert \leq r, \quad
    	L_n = Cn^{\frac{1}{2}}, \quad
    	\Delta_n = n^{-1},
    \end{align*}
    yields the claim
    \begin{align*}
    	\left\lvert \sum_{k=n^{\frac{1}{3}}}^{r'n} \sigma_{C,\phi}\left(-\lvert a\rvert+\frac{k}{n}\right)n^{-1}-\int_{n^{-\frac{2}{3}}}^{r'}\sigma_{C,\phi}(-\lvert a\rvert+x)\d x\right\rvert&\leq \lvert I_n\rvert L_n \Delta_n\leq r C n^{-\frac{1}{2}}.\qedhere
    \end{align*}
\end{proof}
\subsection{Proof of \texorpdfstring{\cref{thm:quantum_rate_supnorm}}{Theorem 2.2}}
\label{subsec:finalproof}

	First of all, by the triangle inequality and linearity of the trace, note that we can restrict our attention to $\rho=\ket{\delta_0\phi}\braket{\delta_0\phi}$.
    The lower bound then follows from \cref{lem:suplowbound} combined with the fact $F_V^\phi(-|a|)=0$, cf. \cref{lem:asymptotic}.
    To prove the upper bound, pick $r>0$ and $C>0$ as in \cref{thm:wallerror}, which especially implies $|F_{X_n}^\phi(x)-F_V^\phi(x)|\le C n^{-1/3}$ for all $x\in[\pm|a|\mp r,\pm|a|]$.
    Further, for $x\in[-|a|+r,|a|-r]$, \cref{lem:inside} yields $C_r>0$ such that $|F_{X_n}^\phi(x)-F_V^\phi(x)|\le C_rn^{-1/3}$.
    Finally, \cref{thm:leftbound} yields $C'>0$ such ath $|F_{X_n}^\phi(x)-F_V^\phi(x)|\le C' n^{-1/3}$ for $|x|\ge |a|$.
    Combining these estimates proves the statement.
\qed

\appendix
\section{Oscillatory Sum Estimates}
\label{app:osc}
In this appendix, we prove estimates on sums over oscillatory terms, which we applied in \cref{sec:wavefront} to bound the supremum-distance in the wavefront region.

Let us start, by recalling the following two statements.
\begin{lem}[{\cite[Thm.~2.1]{GrahamKolesnik.1991}}]\label{prop:landau}
    Let $I$ be an interval and let $f\in C^1(I)$ have monotone derivative, then
    \begin{align*}
        \left\lvert\sum_{k\in\IZ\cap I}\exp(if(k))\right\rvert\leq \Big(\inf_{x\in I}|f'(x)|\Big)^{-1},
    \end{align*}
    where the right hand side is infinite if the infimum is zero.
\end{lem}
\begin{lem}[{\cite[Thm.~2.2]{GrahamKolesnik.1991}}]\label{prop:corput}
    Let $I$ be an interval, let $f\in C^2(I)$ and assume there are $\lambda>0$ and $\gamma\ge 1$ such that
    \begin{align*}
        \lambda\leq \lvert f''(x)\rvert \leq \gamma \lambda \textup{ for all } x\in I.
    \end{align*}
    Then
    \begin{align*}
        \left\lvert\sum_{n\in \IZ\cap I }\exp(if(n))\right\rvert\leq \gamma \lvert I\rvert \lambda^{\frac{1 }{2}}+\lambda^{-\frac{1}{2}}.
    \end{align*}
\end{lem}
From here, we prove the following result on the vanishing speed of the sum over oscillating arguments.
\begin{prop}\label{lem:cancel}
    For an open domain $D$ around zero, let $p\in C^2(D)$ with $p(0)=0$ and $p'(0)=\alpha>0$. Then, there exists $r>0$ and $C>0$ such that, for $n\in\IN$ large enough,
    \begin{align*}
        \left\lvert\sum_{k=\lfloor n^{\frac 1 3 }\rfloor}^{\lfloor rn \rfloor} \sin\left(\tfrac{4}{3}n p\left(\tfrac{k}{n}\right)^{\frac{3}{2}}\right)\right\rvert\leq C n^{\frac{1}{2}}.
    \end{align*}
    Furthermore, for any sequence of subintervals $(I_n)_{n\in\IN}$ with $I_n\subseteq [n^{1/3},rn]$, we have
    \begin{align*}
        \left\lvert\sum_{k\in I_n\cap\IZ} \sin\left(\tfrac{4}{3}n p\left(\tfrac{k}{n}\right)^{\frac{3}{2}}\right)\right\rvert\leq C n^{\frac{1}{2}}.
    \end{align*}
\end{prop}
\begin{proof}
    We define the functions 
    \begin{align*}
        f_n(x)&\coloneq \tfrac{4}{3}n p\left(\tfrac{x}{n}\right)^{\frac{3}{2}}\\
        f_n'(x)&=2a\left(\tfrac{x}{n}\right)p'\left(\tfrac{x}{n}\right)\\
        f_n''(x) &= n^{-1}\left(2p''\left(\tfrac{x}{n}\right)ü\left(\tfrac{x}{n}\right)^{\frac{1}{2}}+\frac{1}{2}p'\left(\tfrac{x}{n}\right)^{2}p\left(\tfrac{x}{n}\right)^{-\frac{1}{2}}\right).
    \end{align*}
    In this proof, we split the sum over $\sin(f_n(k))$ into the regions $k\in[n^{2/3},rn]$ and $k\in[n^{1/3},n^{2/3})$, where we will use \cref{prop:landau} to estimate the first part and \cref{prop:corput} for the second. 

    For the first part, by the assumptions, there exists $r>0$ such that the function $f'_n(n\,\cdot\,)$ is monotonically increasing on $[0,r]$
    and thus for any $n\in\IN$ the function $f'_n$ is monotone on the interval $[n^{2/3},rn)$
    and there exists a $C>0$ such that
    \begin{align*}
        f'_n(x)\geq 2p(n^{-\frac 1 3 })p'(n^{-\frac 1 3 })>Cn^{-\frac 1 3 }.
    \end{align*}
    Then, by \cref{prop:landau}, we obtain
    \begin{align}
    	\label{eq:firstsum}
        \sum_{k=\lfloor n^{\frac 2 3 }\rfloor}^{\lfloor rn\rfloor} \sin\left(\tfrac{4}{3}n p\left(\tfrac{k}{n}\right)^{\frac 3 2 }\right) =O(n^{\frac 1 3 }).
    \end{align}
  	To estimate the second part, we notice that the assumptions on $p$ imply that
    the term including $p(x/n)^{-1/2}$ is monotonically decreasing and dominating the behaviour of $f''_n$, since the other term is bounded, and thus there exists $C>0$ such that, for $x\in[n^{1/3},n^{2/3})$,
    \begin{align*}
        f''_n(x)&\geq f''_n(n^{\frac 2 3 })\geq Cn^{-1}(n^{-\frac 1 3 })^{-\frac 1 2 }=Cn^{-5/6}\geq Cn^{-1}\\
        f''_n(x)&\leq f''_n(n^{\frac 1 3 })\leq Cn^{-1}(n^{-\frac 2 3 })^{-\frac 1 2 }\leq Cn^{-\frac 2 3 }
    \end{align*}
    Thus, with the choice 
    \begin{align*}
        \lambda\coloneq Cn^{-1},\qquad
        \gamma\coloneq n^{\frac 1 3 },\qquad
        \lvert I_n\rvert =n^{\frac 2 3 }, 
    \end{align*}
    \cref{prop:corput} yields the estimate
    \begin{align}\label{eq:secondsum}
        \left\lvert\sum_{k=\lfloor n^{\frac 1 3 }\rfloor}^{\lfloor n^{\frac 2 3 }\rfloor}\sin\left(f_n(k)\right)\right\rvert&\leq \left(Cn^{-1}\right)^{\frac{1}{2}}n^{\frac 1 3 }n^{\frac 2 3 }+\left(Cn^{-1}\right)^{-\frac{1}{2}}= O\left(n^{\frac{1}{2}}\right).  
    \end{align}
    Combining \cref{eq:firstsum,eq:secondsum} yields the claimed order of $O(n^{1/2})$.
    All estimates above are also applicable for any sequence of subintervals $(I_n)_{n\in\IN}$ with $I_n\subset [n^{1/3},rn]$ and the estimates hold with the same constants. 
\end{proof}
\begin{cor}
	\label{cor:cancel}
	Let $p$ and $r>0$ and $C>0$ be as in \cref{lem:cancel}.
	Then, for $f:(0,\infty)\to\IR$ positive and monotone function and any sequence $(s_n)_{n\in\IN}\subseteq \IN$ with $s_n\leq rn$, we have 
	\[ \left\lvert  \sum_{k=s_n}^{\lfloor rn\rfloor}f\left(\tfrac{k}{n}\right)  \sin\left(\tfrac{4}{3}n p\left(\tfrac{k}{n}\right)^{\frac 3 2 }  \right)\right\rvert \le 
	Cn^{\frac{1}{2}}\left(2\left\lvert f\left(\tfrac{\lfloor rn\rfloor +1}{n}\right)\right\rvert +\left\lvert f\left(\tfrac{s_n}{n}\right)\right\rvert\right).
	 \]
\end{cor}
\begin{proof}
	The claim follows from \cref{lem:cancel} and the assumptions on $f$ with the telescope-type argument
	\begin{align*}
		&\left\lvert  \sum_{k=s_n}^{\lfloor rn\rfloor}f\left(\tfrac{k}{n}\right)  \sin\left(\tfrac{4}{3}np\left(\tfrac{k}{n}\right)^{\frac 3 2 }  \right)\right\rvert\\
		&=\left\lvert  f\left(\tfrac{\lfloor rn\rfloor +1}{n}\right) \sum_{k=s_n}^{\lfloor rn\rfloor}\sin\left(\tfrac{4}{3}np\left(\tfrac{k}{n}\right)^{\frac 3 2 }  \right)+\sum_{k=s_n}^{\lfloor rn\rfloor}\left(f\left(\tfrac{k}{n}\right) - f\left(\tfrac{k+1}{n}\right) \right)\sum_{l=s_n}^{k}\sin\left(\tfrac{4}{3}n p\left(\tfrac{k}{n}\right)^{\frac 3 2 }  \right)\right\rvert\\
		&\leq \left\lvert  f\left(\tfrac{\lfloor rn\rfloor +1}{n}\right) \right\rvert\left\lvert\sum_{k=s_n}^{\lfloor rn\rfloor}\sin\left(\tfrac{4}{3}n p\left(\tfrac{k}{n}\right)^{\frac 3 2 }  \right)\right\rvert+\sum_{k=s_n}^{\lfloor rn\rfloor}\left\lvert  f\left(\tfrac{k}{n}\right) - f\left(\tfrac{k+1}{n}\right) \right\rvert\left\lvert\sum_{l=s_n}^{k}\sin\left(\tfrac{4}{3}n a\left(\tfrac{k}{n}\right)^{\frac 3 2 }  \right)\right\rvert\\
		&\leq Cn^{\frac{1}{2}}\left(\left\lvert f\left(\tfrac{\lfloor rn\rfloor +1}{n}\right)\right\rvert+ \sum_{k=s_n}^{\lfloor rn\rfloor}\left\lvert  f\left(\tfrac{k}{n}\right) - f\left(\tfrac{k+1}{n}\right) \right\rvert\right) \\
		&\leq Cn^{\frac{1}{2}}\left(2\left\lvert f\left(\tfrac{\lfloor rn\rfloor +1}{n}\right)\right\rvert +\left\lvert f\left(\tfrac{s_n}{n}\right)\right\rvert\right) .
		\qedhere
	\end{align*}
\end{proof}
From this general result, we turn to the following estimate important for our application.
\begin{prop}\label{lem:OSCest2} Let the function $p$ be smooth on an open set domain around zero with $p(0)=0$ and $p'(0)=\alpha>0$, then 
    \begin{align*}
        n^{-1}\sum_{k=n^{\frac{1}{3}}}^{n^{\frac{2}{3}}}p\left(\tfrac{k}{n}\right)^{-\frac{1 }{2}} \sin\left(\tfrac{4}{3}np\left(\tfrac{k}{n}\right)^{\frac{3}{2}}\right)= O(n^{-\frac{1}{3}})
    \end{align*}
\end{prop}
\begin{proof}
    We want to identify the sum as Riemann sum and compare it to the corresponding integral 
    \begin{align*}
        \int_{n^{-\frac 2 3 }}^{n^{-\frac 1 3 }}p(x)^{-\frac 1 2 }\sin\left(\tfrac{4}{3}np\left(x\right)^{\frac 3 2 }\right)\d x.
    \end{align*}
    This integral exists as an integral of smooth functions over a bounded interval. 
    To arrive at the stated estimate we therefore need to evaluate the integral and estimate the approximation error. 

    With the identifications
    \begin{align*}
        f(x)&\coloneqq p(x)^{-\frac{1}{2}}\sin\left(\tfrac{4}{3}np(x)^{\frac{3}{2}}\right),\\
        f'(x)&=-\frac{1}{2}p'(x)p(x)^{-\frac{3}{2}}\sin\left(\tfrac{4}{3}np(x)^{\frac{3}{2}}\right)+2np'(x)\cos\left(\tfrac{4}{3}np(x)^{\frac{3}{2}}\right),\\
        L&\coloneq \sup_{n^{-\frac{2}{3}}\leq x\leq n^{-\frac{1}{3}}}f'(x)=O(n),\\
        \lvert I\rvert &\coloneq n^{-\frac{1}{3}}-n^{-\frac{2}{3}}\leq n^{-\frac{1}{3}},\\
        \Delta x &= n^{-1},
    \end{align*}
    we can apply the approximation error formula to estimate
    \begin{align*}
        \left\lvert n^{-1}\sum_{k=n^{\frac{1}{3}}}^{n^{\frac{2}{3}}}p\left(\tfrac{k}{n}\right)^{-\frac{1 }{2}} \sin\left(\tfrac{4}{3}np\left(\tfrac{k}{n}\right)^{\frac{3}{2}}\right)-\int_{n^{-\frac{2}{3}}}^{n^{-\frac{1}{3}}}p(x)^{-\frac{1}{2}}\sin\left(\tfrac{4}{3}np(x)^{\frac{3}{2}}\right)\d x\right\rvert\leq L\lvert I\rvert \Delta x = O(n^{-\frac{1}{3}}).
    \end{align*}
    It thus remains to evaluate the integral.
    First, notice that there exists $r>0$ such that $p$ is injective on the interval $[-r,r]$.
    Thus, in the following we can choose $n$ large enough such that $n^{-1/3}\leq r$.
    Then, we apply the substitution $u=p(x)$ to obtain
    \begin{align*}
        \int_{n^{-\frac 2 3 }}^{n^{-\frac 1 3 }}p(x)^{-\frac 1 2 }\sin\left(\tfrac{4}{3}np\left(x\right)^{\frac 3 2 }\right)\d x&=\int_{p^{-1}(n^{-\frac{2}{3}})}^{p^{-1}(n^{-\frac{1}{3}})}u^{-\frac{1}{2}}\sin\left(\tfrac{4}{3}nu^{\frac{3}{2}}\right)    \frac{\d u}{p'(p^{-1}(u))}.
    \end{align*}
    Noting that by the inverse function theorem the derivative
    \begin{align*}
        \left(\frac{1}{p'(p^{-1}(u))}\right)'=-\frac{p''(p^{-1}(u))}{\left(p'(p^{-1}(u))\right)^3}
    \end{align*}
    is smooth on the compact interval $u\in p([-r,r])$, there exists $C>0$ such that 
    \begin{align*}
        \sup_{u\in p([-r,r])}\left\lvert\frac{p''(p^{-1}(u))}{\left(p'(p^{-1}(u))\right)^3}\right\rvert\le C. 
    \end{align*}
    With this, we can estimate the error that arises from linearizing the term $(p'(p^{-1}(u)))^{-1}$ in the integral, i.e.,
    \begin{align*}
        \left\lvert \int_{p^{-1}(n^{-\frac{2}{3}})}^{p^{-1}(n^{-\frac{1}{3}})}u^{-\frac{1}{2}}\sin\left(\tfrac{4}{3}nu^{\frac{3}{2}}\right)\left(    \frac{1}{p'(p^{-1}(u))} -\alpha^{-1}\right)\d u\right\rvert&\leq \int_{a^{-1}(n^{-\frac{2}{3}})}^{a^{-1}(n^{-\frac{1}{3}})}u^{-\frac{1}{2}}\left\lvert    \frac{1}{p'(p^{-1}(u))} -\alpha^{-1}\right\rvert\d u\\
        &\leq  \int_{p^{-1}(n^{-\frac{2}{3}})}^{p^{-1}(n^{-\frac{1}{3}})}u^{-\frac{1}{2}}Cu\d u\\
        &\leq \frac{2}{3}C \left(p^{-1}(n^{-\frac 1 3 })\right)^{\frac 3 2 }\\
        &=O(n^{-\frac 1 2 }).
    \end{align*}
    Then evaluating the integral with the linearized integrand  with the substitution $nx^{3/2}\mapsto x$ yields
    \begin{align*}
        \int_{n^{-\frac{2}{3}}}^{n^{-\frac{1}{3}}}p(x)^{-\frac{1}{2}}\sin\left(\tfrac{4}{3}np(x)^{\frac{3}{2}}\right)\d x&=\int_{p^{-1}(n^{-\frac{2}{3}})}^{p^{-1}(n^{-\frac{1}{3}})}x^{-\frac{1}{2}}\sin\left(\tfrac{4}{3}nx^{\frac{3}{2}}\right)    \alpha^{-1}\d x +O(n^{-\frac 1 2 })\\
        &= \int_{n(p^{-1}(n^{-\frac{2}{3}}))^{\frac{3}{2}}}^{n(p^{-1}(n^{-\frac{1}{3}}))^{\frac{3}{2}}}\left(\tfrac{x}{n}\right)^{-\frac{1}{3}}n^{-\frac{2}{3}}x^{-\frac{1}{3}}\sin(\tfrac{4}{3}x)\alpha^{-1}\d x +O(n^{-\frac{1}{2}})\\
        &=n^{-\frac{1}{3}}\alpha^{-1}\int_{n(a^{-1}(n^{-\frac{2}{3}}))^{\frac{3}{2}}}^{n(p^{-1}(n^{-\frac{1}{3}}))^{\frac{3}{2}}}x^{-\frac{2}{3}}\sin(x)\d x +O(n^{-\frac{1}{2}})\\
        &=O(n^{-\frac{1}{3}}),
    \end{align*}
    since the integral on the right is bounded by the convergence of 
    \begin{align*}
        \int_{0}^{\infty} x^{-\frac{2}{3}}\sin(x)\d x =\sqrt{2\pi}.
    \end{align*}
	Combining all above observations proves the statement.
\end{proof}
\section{Stationary Phase Method}
\label{app:stationary}
In this appendix, we sketch the proof of \cref{thm:tate.2} from the arguments given in \cite[\S\,3]{SunadaTate.2012}.
Throughout, we drop the index $\pm$ in $p$, $q$ and $s$.

Starting point is the following modification of \cite[Prop.~3.1]{SunadaTate.2012}.
\begin{prop}\label{prop:TateAdaptation}
	In a neighborhood around zero there exist smooth functions $p,f,q_j,s_j$ for $j\in\{1,2\}$ with $p(0)=0$ and $p'(0)=\alpha>0$ and $r>0$ such that for any $y_n(k)\coloneq \pm\lfloor\lvert a\rvert n\rfloor\mp k$ with $n,k\in\IN$ such that $k/n\in ( -r,r)$, we have
	\begin{align*}
		\braket{W^n\phi,\delta_{y_n(k)} e_j}
		&= e^{i\pi(n\mp k)/2}e^{inf \left(\frac{k}{n}\right)}\left(n^{-\frac{1}{3}}s_j \left(\tfrac{k}{n}\right)\textup{Ai}\left(\pm n^{\frac{2}{3}}p \left(\tfrac{k}{n}\right)\right)-in^{-\frac{2}{3}}q_j \left(\tfrac{k}{n}\right)\textup{Ai}'\left(\pm n^{-\frac{1}{3}}p \left(\tfrac{k}{n}\right)\right)\right)
		\\&\qquad +O(n^{-1}), \qquad n\to\infty.
	\end{align*} 
\end{prop}
\begin{proof}
	From \cite[Eq.~(2.11)]{SunadaTate.2012}, setting $w\coloneq a/\lvert a\rvert $, we know
	\begin{align*}
		\braket{W^n\delta_0\phi|\delta_{y_n(k)}e_j}=w^{-y_n(k)}(1+(-1)^{n+y_n(k)})J(e_j,n,y_n(k))
	\end{align*}
	with $J$  given by \cite[p.~2620]{SunadaTate.2012}
	\begin{align*}
		J(e_j,n,y_n(k)) = J_1(e_j,n,y_n(k)) + O(n^{-\infty}), \quad n\to\infty,
	\end{align*}
	where $J_1$ as in \cite[Eqs.~(3.10),(3.11)]{SunadaTate.2012}, for a function $\chi\in C^{\infty}_0(\IR)$ with $0\leq \chi\leq1$, satisfies
	\begin{align*}
		&J_1(e_j,n,y_n(k))\\&=\frac{e^{i\pi(n-y_n(k))}}{2\pi}e^{inf(\frac{k}{n})}\bigg(s_j\left(\tfrac{k}{n}\right)n^{-\frac 1 3 }\left(2\pi\textup{Ai}\left(\pm n^{\frac 2 3 }p\left(\tfrac{k}{n}\right)\right)-\int e^{in(t^3/3+p(k/n)t)}(1-\chi(n^{-\frac 1 3 }t))\d t\right)\\
		&\quad-q_j\left(\tfrac{k}{n}\right)n^{-\frac 2 3 }\left(2\pi i\textup{Ai}'\left(\pm n^{\frac 2 3 }p\left(\tfrac{k}{n}\right)\right)+\int e^{in(t^3/3+p(k/n)t)}t(1-\chi(n^{-\frac 1 3 }t))\d t\right)\bigg) + O(n^{-1})
	\end{align*}
	To estimate the integrals, we then apply \cite[Lemma~3.2]{SunadaTate.2012},
	which with the substitution $s=n^{-1/3}t$ yields
	\begin{align*}
		J_1(e_j,n,y_n(k))&=e^{i\pi(n-y_n(k))}e^{inf(\frac{k}{n})}\bigg(n^{-\frac 1 3 }s_j\left(\tfrac{k}{n}\right)\textup{Ai}\left(\pm n^{\frac 2 3 }p\left(\tfrac{k}{n}\right)\right)-q_j\left(\tfrac{k}{n}\right)n^{-\frac 2 3 } i\textup{Ai}'\left(\pm n^{\frac 2 3 }p\left(\tfrac{k}{n}\right)\right)\bigg) \\&\qquad + O(n^{-1}).
	\end{align*}
	Combining these observations proves the claim.
\end{proof}
We conclude with the
\begin{proof}[Proof of \cref{thm:tate.2}] We will continue in the setting of the previous proof and first show $$ J_1(e_j,n,y_n(k))=O(n^{-1/3}).$$
To this end, note that continuity of $p$ as well as 
$\operatorname{Ai}'(x)\le C(1+x^{1/4})$ \cite[\S\,11.1.4]{Olver.1974} imply
	\begin{align*}
		\left\lvert n^{-\frac 1 3 }\textup{Ai}'\left(\pm n^{\frac 2 3 }p(x)\right)\right\rvert
		\leq Cn^{-\frac 1 3 } + Cp(x)^{1/4}n^{-1/6}.
	\end{align*}
	Therefore, we see
	\begin{align*}
		\left\lvert J_1(e_j,n,y_n(k))\right\rvert &\leq n^{-\frac 1 3 }\left\lvert s_j\left(\tfrac{k}{n}\right)\textup{Ai}\left(\pm n^{\frac 2 3 }p\left(\tfrac{k}{n}\right)\right)-q_j\left(\tfrac{k}{n}\right)n^{-\frac 1 3 } i\textup{Ai}'\left(\pm n^{\frac 2 3 }p\left(\tfrac{k}{n}\right)\right)\right\rvert + O(n^{-1}) \\
		&=O(n^{-\frac 1 3 })
	\end{align*}
	Combined with \cref{prop:TateAdaptation}, this yields
	\begin{align*}
		&\left\lvert \braket{W^n(\delta_0\phi)|\delta_{y_n(k)}e_j}\right\rvert^2=(1+(-1)^{n+y_n(k)})J_1(e_j,n,y_n(k))(1+(-1)^{n+y_n(k)})\bar{J_1(e_j,n,y_n(k))}+O(  {n^{-\frac 4 3 }})\\
		&\quad=(1+(-1)^{n+y_n(k)})^2\left(s_j\left(\tfrac{k}{n}\right)n^{-\frac 1 3 }\textup{Ai}\left(n^{\frac 2 3 }p\left(\tfrac{k}{n}\right)\right)-iq_j\left(\tfrac{k}{n}\right)n^{-\frac 2 3 }\textup{Ai}'\left(n^{\frac 2 3 }p\left(\tfrac{k}{n}\right)\right)\right)\\
		&\quad \qquad \times \left(s_j\left(\tfrac{k}{n}\right)n^{-\frac 1 3 }\textup{Ai}\left(\pm n^{\frac 2 3 }p\left(\tfrac{k}{n}\right)\right)+iq_j\left(\tfrac{k}{n}\right)n^{-\frac 2 3 }\textup{Ai'}\left(\pm n^{\frac 2 3 }p\left(\tfrac{k}{n}\right)\right)\right)+O(  {n^{-\frac 4 3 }})\\
		&\quad =(1+(-1)^{n+y_n(k)})^2\left(s_j^2\left(\tfrac{k}{n}\right)n^{-\frac 2 3 }\textup{Ai}^2\left(\pm n^{\frac 2 3 }p\left(\tfrac{k}{n}\right)\right)+q_j^2\left(\tfrac{k}{n}\right)n^{-\frac 4 3 }\textup{Ai}'^2\left(\pm n^{\frac 2 3 }p\left(\tfrac{k}{n}\right)\right)\right)\\
		&\quad\qquad +O(  {n^{-\frac 4 3 }}).
	\end{align*}
	Summing over $j=1,2$ and setting $q^2\coloneq q_1^2+q_2^2$ and $s^2\coloneq s_1^2+s_2^2$ yields the claim.
\end{proof}

\bibliographystyle{halpha-abbrv}
\bibliography{../../../CLOUD/Literature/00lit}
\end{document}